\documentclass[lettersize,journal]{IEEEtran}
\IEEEoverridecommandlockouts
\usepackage{cite}
\usepackage{amsmath,amssymb,amsfonts}
\usepackage{algorithmic}
\usepackage{graphicx}
\usepackage{textcomp}
\usepackage{xcolor}

\newtheorem{theorem}{Theorem}
\newtheorem{lemma}{Lemma}
\newtheorem{remark}{Remark}

\newtheorem{proof}{Proof}
\def\BibTeX{{\rm B\kern-.05em{\sc i\kern-.025em b}\kern-.08em
		T\kern-.1667em\lower.7ex\hbox{E}\kern-.125emX}}

\usepackage{stfloats}
\usepackage{bbm}
\usepackage{bm}

\ifCLASSINFOpdf
\else
\fi
\begin{document}
	%
	\title{Faster-than-Nyquist Signaling \\ for MIMO Communications}
	%
	%
	%
	
	\author{Zichao~Zhang,~\IEEEmembership{Student Member,~IEEE,}
		Melda~Yuksel,~\IEEEmembership{Senior Member,~IEEE,}
		and~Halim~Yanikomeroglu,~\IEEEmembership{Fellow,~IEEE}
		\thanks{This work was supported in
part by the Natural Sciences and Engineering Research Council of
Canada, NSERC, under a Discovery Grant, in part by the Scientific and
Technological Research Council of Turkey, TUBITAK, under Grant
122E248, and in part by Nokia-Bell Labs, under the Corporate
University Donations Program.
}\thanks{Z. Zhang and H. Yanikomeroglu are with the Department of Systems and Computer Engineering at Carleton University, Ottawa, ON, K1S 5B6, Canada e-mail: zichaozhang@cmail.carleton.ca,	halim@sce.carleton.ca.}
		\thanks{M. Yuksel is with the Department of Electrical and Electronics Engineering, Middle East Technical University, Ankara, 06800, Turkey, e-mail: ymelda@metu.edu.tr.}}
	\maketitle
	
	\begin{abstract}
		Faster-than-Nyquist (FTN) signaling is a non-orthogonal transmission technique, which has the potential to provide significant spectral efficiency improvement. This paper studies the capacity of FTN signaling for both frequency flat and for frequency selective (FS) multiple-input multiple-output (MIMO) channels. {As FTN is another reason of frequency selectivity, we find that precoding in time (or equivalently spectrum shaping in frequency) and waterfilling in spatial  domain is capacity achieving for  frequency flat  MIMO channels with FTN. Meanwhile, waterfilling in both spatial  domain and spectrum domain, followed by spectrum shaping, is capacity achieving for FS MIMO channels with FTN.}

	\end{abstract}
	

	\section{Introduction}\label{sec:intro}
	
	In recent years, there has been an increasing interest in faster-than-Nyquist (FTN) signaling, as researchers are seeking for new breakthrough technologies to improve transmission rates in modern wireless communication systems. FTN signaling happens to be a promising solution for this search \cite{5gand6g} because of the higher spectral efficiency it offers. 
	
	
	In classical digital communications, orthogonal pulses or Nyquist pulses are used so that there is no inter-symbol interference (ISI) at the demodulator output at sampling time instants. This maximum signaling rate above which orthogonality between consecutive symbols no longer exist is called the Nyquist limit.
	FTN, on the other hand, intentionally operates beyond this Nyquist limit. It introduces intentional ISI using the same pulse shape within the same bandwidth to improve spectral efficiency. In \cite{mazo}, Mazo studied a binary linear modulation scheme utilizing sinc pulses, where the Nyquist limit is $\frac{1}{T}$. He showed that it is possible to decrease the symbol interval down to 0.802$T$ with 
	no performance degradation. 
	This rate is called the Mazo's limit, it is the limiting rate beyond which performance starts to deteriorate for the same signal-to-noise ratio (SNR) for the same bandwidth. 
	
 	A reasonable question at this point is: What is the maximum transmission rate over an FTN channel? For orthogonal communication over additive white Gaussian noise (AWGN) channels, the capacity expression for a bandpass channel is the one obtained by Shannon \cite{Cover}, which states that
\begin{eqnarray}
			C_{orth} & = & \frac{1}{2}\log_2\left(1+\frac{P}{WN_0}\right) ~\text{bit/s/Hz}, \label{eqn:nonFTNcap} 
	\end{eqnarray}when $P$ is the power constraint, $N_0/2$ is the AWGN power spectral density, $W=\frac{1}{2T}$ is the bandwidth, and $\frac{1}{T}$ is the signaling rate. 
	In the literature, there are multiple papers that seek for a similar capacity expression for a single-input single-output (SISO) FTN system. In \cite{rusek,property}, the authors find an achievable rate for a SISO FTN system assuming that transmitted symbols are independent. 
	The optimal input symbol distribution, on the other hand, is not necessarily independent, and some correlation can help reduce ISI. In \cite{bajcsy} and \cite{rusek12}, correlation among input symbols is explored and a power allocation scheme based on waterfilling is suggested. However, because of the influence of other symbols, the power constraint of FTN signaling differs from that of orthogonal signaling, and the power constraints in \cite{bajcsy} or \cite{rusek12} do not account for the effect of ISI. In \cite{spectrashpingKramer,ganji} the authors include the impact of ISI on the power constraint and show that the optimal power spectral density (PSD) of the transmitted symbols is related to the spectrum of the pulse used for signaling. 
{Additionally, \cite{timelocalization} argues that for the capacity expression to be valid, the resultant FTN signal must be well-localized in time, meaning that it should contain a significant portion of its energy in the corresponding time interval.}  
In addition to the above, precoding at the transmitter side is a viable method to mitigate FTN induced ISI \cite{linearprecoding}. Similar to \cite{linearprecoding}, \cite{svd, eigendecomposition} and \cite{chaki} suggest an eigen-decomposition based precoding for FTN signaling. 


	It is well known that using multiple antennas at both transmitter and receiver sides in wireless communications can significantly improve the channel capacity \cite{telatar}. Therefore, the multiple-input multiple-output (MIMO) technology has been an integral part of 3GPP standards since Release 7 \cite{3gpp}. In order to further improve communication rates, it is then reasonable to employ FTN signaling in MIMO systems. Rusek \cite{rusek2009existence} studied the Mazo limit in MIMO communication systems. Modenini, et al. \cite{andrea} demonstrated performance improvements for FTN in large-scale antenna systems with a simple receiver structure. Yuhas, et al. \cite{michael} explored MIMO FTN capacity with independent inputs as in \cite{rusek}, but their assumption is not general enough as the optimal input distribution is not necessarily independent. {The ergodic capacity of MIMO FTN over triply-selective Rayleigh fading channels is studied in \cite{MIMOFTNfad} with the assumption that channel state information is not available at the transmitter.} However, the capacity of a MIMO system with FTN signaling with channel state information at both the transmitter and the receiver is still unknown.
	
	{Frequency selectivity is another important phenomenon in fading channels \cite{goldsmith}, and 
	many results on frequency selective (FS) MIMO channels have appeared in the literature. 
	In \cite{measurecap}, the authors achieved the capacity of FS MIMO channels with joint waterfilling in spatial and frequency domains. 
	In \cite{covmat}, the authors proposed an algorithm to evaluate the capacity-achieving covariance matrix for FS MIMO channels based on an iterative waterfilling scheme. 
In \cite{finite}, a method is proposed to derive the average capacity of FS MIMO channels for arbitrary finite number of transmit and receive antennas. 
	There are also efforts on exploring the performance of FTN in FS channels.
	The authors of \cite{eigendecomposition} obtain the maximized mutual information for a FS SISO channel with FTN as well as its optimal input distribution.
	In \cite{wangjui}, a time-space transformation approach is proposed to increase the achievable rate in FS MIMO channels with FTN. Finally, \cite{simo} suggests an FTN transmission and detection approach for FS SISO systems. However, to the best of our knowledge there is no result on the capacity of FS MIMO channels with FTN. In this paper we study the channel capacity of MIMO FTN signaling both for frequency flat and for FS fading channels.}

	
	The organization of the paper is as follows. For MIMO FTN, we define the system model in Section~\ref{sec2}. We develop the mutual information expression in Section~\ref{subsec3A}, find the related power constraint in Section~\ref{subsec3B}, and solve for the capacity expression in Section~\ref{subsec3C}. In Section~\ref{sec:freqsel}, we generalize the results for MIMO FTN to FS channels. In Section~\ref{sec:freqdomanalysis}, we calculate the same capacity via frequency domain analysis. 
	We present our numerical results in Section~\ref{sec8}. Finally, in Section~\ref{sec9} we conclude the paper.

	\section{System Model }
	\label{sec2}
	
	In MIMO FTN, each transmitting antenna transmits symbols simultaneously using FTN signaling. We assume the transmitter has  $L$, and the receiver has  $K$ antennas. The MIMO channel matrix is denoted with  $\bm{\tilde{H}}\in \mathbb{C}^{K \times L}$. The  $(k,l)$th entry, $h_{kl}$, in $\bm{\tilde{H}}$ is the complex channel gain from transmit antenna $l$ to receive antenna $k$, $l \in \{1,\cdots, L\}$, $k \in \{1,\cdots, K\}$.
  The  sequence of  data symbols transmitted from the $l$th transmitting antenna $a_l[n]$ passes through the pulse shaping filter $p(t)$. We assume that the pulse shaping filter has unit energy; i.e., $\int_{-\infty}^{+\infty}|p(t)|^2dt=1.$ Moreover, the transmitter has a power constraint $P$. The transmitted signal from this antenna for a length of $N$ symbols is written as 
	\begin{align}
x_l(t)&=\sum_{m=1}^{N}a_l[m]p(t-m\delta T),\label{eqn:xt}
	\end{align}
	where $T$ denotes the signaling period, and $\delta \in (0,1]$ is the acceleration (compression) factor in FTN signaling. In FTN, pulses are transmitted every $\delta T$ seconds, creating intentional inter-symbol interference. 

	We assume an additive white Gaussian noise (AWGN) channel and denote the noise at the $k$th receiving antenna with  $\xi_k(t), k=1,2,\cdots,K$. This noise process has zero mean and a constant power spectral density equal to $\sigma_0^2$. We assume noise at each receiving antenna are independent and identically distributed and are independent of data symbols transmitted.
	In order to maximize the signal-to-noise ratio (SNR),  matched filter $p^*(-t)$, where $*$ denotes the complex conjugate, is used at the receiver side. 
	Then, the output of the matched filter at the $k$th receive antenna can be written as
	\begin{align}
		y_k(t) &= \left[\sum_{l=1}^{L}h_{kl}\left(\sum_{m=1}^{N}a_l[m]p(t-m\delta T)\right) + \xi_k(t)\right] \star p^*(-t),  \label{eqn:received} \\
		&=\sum_{l=1}^{L}h_{kl}\left(\sum_{m=1}^{N}a_l[m]g(t-m\delta T)\right)+\eta_k(t),
	\end{align} where $g(t) \triangleq p(t)\star p^*(-t)$, $\eta_k(t) \triangleq \xi_k(t) \star p^*(-t)$ and $\star$ denotes the convolution operation. This matched filter output at each receive antenna $k$ is then sampled every $\delta T$ seconds. Fig. \ref{fig3} shows the structure of a $2 \times 2$ MIMO FTN system as a simple example.

	\begin{figure}[t]
		\centering
		\includegraphics[scale=0.08]{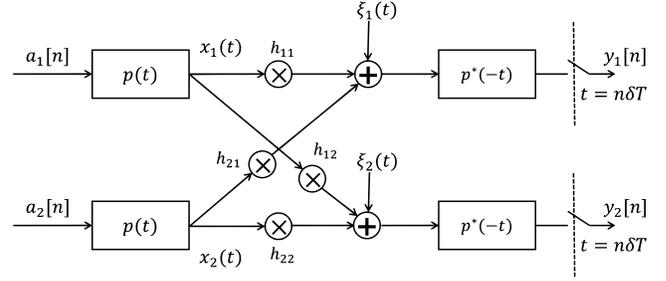}
		\caption{2$\times$2 MIMO FTN system model, where the matched filter outputs are sampled every $\delta T$ seconds.}
		\label{fig3}
	\end{figure}
	
	The noise process $\eta_k(t)$ is a correlated Gaussian process, as the white noise process $\xi_k(t)$ is filtered through $p^*(-t)$. As $\eta_k(m\delta T) = \int_{-\infty}^{\infty}\xi_k(\tau)p^*(\tau-m\delta T)d\tau$, we have
	\begin{align}
	&\mathbb{E}[\eta_k(n\delta T)\eta_k^*(m\delta T)] \notag\\ &=\int_{-\infty}^{\infty}\int_{-\infty}^{\infty}\mathbb{E}[\xi_k(t)\xi_k^*(\tau)]p^*(t-n\delta T)p(\tau-m\delta T)dtd\tau\label{eq5}\\
		&=\sigma_0^2\int_{-\infty}^{\infty}p(t)p^*(t-(n-m)\delta T)dt \\
		&=\sigma_0^2g((n-m)\delta T). \label{eqn:noisecor}
	\end{align}
	Here $\mathbb{E}$ denotes the expectation operation. After sampling, the output of the sampler becomes 
	\begin{align}
		y_k[n] &= \sum_{l=1}^{L}h_{kl}\left(\sum_{m=1}^{N}a_l[m]g((n-m)\delta T) \right) + \eta_k[n].\label{eqn:discretech}
	\end{align} 
	The input-output relationship in \eqref{eqn:discretech} can also be written in a matrix multiplication form as
	\begin{align}
	\bm{y}_k &= \sum_{l=1}^{L}h_{kl}\bm{G}\bm{a}_l+\bm{\eta}_k, \label{eqn:yk}
	\end{align}
	where $\bm{y}_k\in\mathbb{C}^{N\times1}$,  $\bm{a}_l\in\mathbb{C}^{N\times1}, \bm{\eta}_k\in\mathbb{C}^{N\times1}$ and $\bm{G}\in\mathbb{C}^{N\times N}$. The $n$th  entry of the vectors are $(\bm{y}_k)_n=y_k[n]$, $(\bm{a}_l)_n=a_l[n]$ and $(\bm{\eta}_k)_n=\eta_k[n]$, and the $(n,m)$th entry of the matrix $\bm{G}$ is given as $(\bm{G})_{n,m}=g[n-m] \triangleq g((n-m)\delta T)$. Note that $\mathbf{G}$ is a Toeplitz matrix \cite{gray}. The matrix $\bm{G}$ is also Hermitian, as $g(-t) = p(-t) \star p^*(t)=g^*(t)$, and $(\bm{G})_{n,m}=g[n-m] =g^*[m-n]=(\bm{G})^*_{m,n}$. 
From this we observe that $\bm{G}$ is a Gramian matrix as well as a Toeplitz matrix. Using the cross correlation calculated in \eqref{eqn:noisecor}, the covariance matrix for $\bm{\eta}_k, k=\{1,2,\cdots,K\}$, becomes $\bm{\Sigma}_{\bm{\eta}_k}=\mathbb{E}[\bm{\eta}\bm{\eta}^\dagger]=\sigma_0^2\bm{G}$, where $\dagger$ denotes Hermitian conjugation. Manipulating \eqref{eqn:yk} further, we can write the MIMO FTN signaling model as
	\begin{eqnarray}
		\left[\begin{matrix}
			\bm{y}_1 \\ \vdots \\ \bm{y}_K
		\end{matrix}\right] & =&  	\left[\begin{matrix}
			h_{11}\bm{G}& \ldots &h_{1L}\bm{G} \\
			\vdots & \ddots & \vdots \\
			h_{K1}\bm{G}& \ldots & h_{KL}\bm{G}
		\end{matrix}\right] 	\left[\begin{matrix}
			\bm{a}_1 \\ \vdots \\ \bm{a}_L
		\end{matrix}\right] + 	\left[\begin{matrix}
			\bm{\eta}_1 \\ \vdots \\ \bm{\eta}_K
		\end{matrix}\right] \end{eqnarray}\begin{eqnarray}
		\quad\bm{Y}& =& (\bm{\tilde{H}}\otimes\bm{G})\bm{A}+\bm{\Omega} \\
		\quad\bm{Y}& =& \bm{H}\bm{A}+\bm{\Omega}, \label{eqn:MIMOmodel}
	\end{eqnarray} where $\bm{y}_k\in\mathbb{C}^{N\times1}$,  $\bm{a}_l\in\mathbb{C}^{N\times1}$, $\bm{\eta}_k\in\mathbb{C}^{N\times1}$, $\bm{Y} = [\bm{y}_1^T,\cdots \bm{y}_K^T]^T$, $\bm{A} = [\bm{a}_1^T,\cdots \bm{a}_L^T]^T$, $\bm{\Omega} = [\bm{\eta}_1^T,\cdots \bm{\eta}_K^T]^T$, $\bm{H}=\bm{\tilde{H}}\otimes\bm{G}$ and $\otimes$ is the Kronecker product.

Before we proceed with the next section, we have the following lemmas for the positive definiteness of $\bm{G}$. 
	\begin{lemma}\cite{matrixbook}\label{lem:Gposemidef}
		As the matrix $\bm{G}$ is Gramian, it is always positive semi-definite.
	\end{lemma}
	\begin{proof}
		 The proof is provided in \cite{matrixbook}.
	\end{proof}
\begin{lemma}\cite[Proposition 5.1.1]{timeseries}\label{proposition511}
    If $g[0]>0$ and $g[N]\rightarrow0$ as $N\rightarrow\infty$ then the covariance matrix $\bm{\Sigma}_{\bm{\eta}_k}=\sigma_0^2g[i-j]_{i,j=1,\cdots,N}$ of $(\bm{\eta}_k[1],\cdots,\bm{\eta}_k[N])^T$ is non-singular for every $n$, for all $k$.
\end{lemma}
\begin{proof}
        The proof is provided in \cite[Proposition 5.1.1]{timeseries}.
\end{proof}
If root raised cosine pulses with roll-off factor $\beta$ are used for $p(t)$, then $g(t)$ becomes a raised cosine pulse with the frequency response
	\begin{align*}
		G(f)=\left\{\begin{array}{ll}
			T, &|f|\leq\frac{1-\beta}{2T} \\
			\frac{T}{2}\left[1+\cos\left(\frac{\pi T}{\beta}\left(|f|-\frac{1-\beta}{2T}\right)\right)\right], &
		 \frac{1-\beta}{2T}<|f|\leq\frac{1+\beta}{2T}\\
			0, &\text{otherwise}
		\end{array}\right.,  
	\end{align*}
	which satisfies Lemma~\ref{proposition511} \cite{property}. Thus, we also have the following lemma.
{\begin{lemma}\cite{property,timelocalization}\label{lem:Gposdef}
		For a raised cosine pulse $g(t)$ with roll-off factor $\beta$, the matrix $\bm{G}$ is positive definite. However, when $\delta(1+\beta) < 1$, calculating $\bm{G}^{-1}$ is not numerically stable. 
	\end{lemma}}
	\begin{proof}
		The proofs are provided in \cite{property,  timelocalization, eigendecomposition}. 
	\end{proof}

	\section{MIMO FTN Channel Capacity} \label{sec3}

In this section, we find the capacity expression for MIMO FTN transmission in frequency flat channels by using the equivalent discrete time model given in \eqref{eqn:MIMOmodel}. In the following subsection, we derive the related mutual information expression for MIMO FTN. Then, we will determine the power constraint and the capacity. 	
	
	\subsection{Mutual Information Expressions}\label{subsec3A}
Given the discrete time channel model in  \eqref{eqn:MIMOmodel}, we can write the channel capacity as \cite{Cover} 
	\begin{eqnarray}
		C &=& \lim_{N\rightarrow\infty}\max_{p(\bm{A}), \mathcal{P}}\frac{1}{N}I(\bm{Y};\bm{A}), \label{eqn:capacity} \\
		&=&\lim_{N\rightarrow\infty}\max_{p(\bm{A}), \mathcal{P}}\frac{1}{N}[h(\bm{Y})-h(\bm{\Omega})], \label{eqn:capacity2}
	\end{eqnarray} where $I(\bm{Y};\bm{A})$ denotes the mutual information between two random vectors $\bm{Y}$ and $\bm{A}$. The terms $h(\bm{Y})$ and $h(\bm{\Omega})$ respectively denote the differential entropy for $\bm{Y}$ and $\bm{\Omega}$, and $\mathcal{P}$ 
	is the power constraint, which we will explain in detail in Section~\ref{subsec3B}. 
	
	 Because the circularly symmetric complex Gaussian noise process on each receiving antenna ${\xi}_k(t)$ is independent of each other and also independent of input data, 
	samples of $\eta_i(t)$ and $\eta_j(t)$ are also independent; i.e.,  $\mathbb{E}[\eta_i[n]\eta_j^*[k]]=0, i\neq j, \forall k, n \in\mathbb{Z}$. We can then write the covariance matrix and the differential entropy for $\bm{\Omega}$ as 
		\begin{eqnarray}
		\bm{\Sigma_\Omega}&= & \sigma_0^2(\bm{I}_K\otimes\bm{G}), \label{eqn:SigmaOmega}\\
		h(\bm{\Omega})&=&\log_2\left[\left(2\pi e\right)^{KN}\left(\det(\bm{\Sigma_\Omega})\right)\right] , 
	\end{eqnarray} where $e$ is Euler's number, $\bm{I}_K$ is the $K$ by $K$ identity matrix and $\det(.)$ stands for the determinant of a matrix.
	
	\begin{lemma}\cite{matrix}\label{lem:kroneig}
		Let $\bm{\tilde{A}}\in \mathbb{C}^{A \times A} $ and $\bm{\tilde{B}} \in \mathbb{C}^{B \times B}$ be two Hermitian matrices respectively with the eigenvalues $\eta_1,\cdots,\eta_A$ and $\mu_{1},\cdots,\mu_{B}$. Then the matrix $\bm{\tilde{A}}\otimes\bm{\tilde{B}}$ has the eigenvalues $\eta_i\mu_j$, for $i=1,2,\cdots,A$ and $j=1,2,\cdots,B$.
	\end{lemma}
	\begin{proof}
		The proof is provided in \cite{matrix}.
	\end{proof}
	
	 \begin{remark}
	    The eigenvalues of the identity matrix $\bm{I_K}$ are all 1s. Then, according to Lemma~\ref{lem:kroneig}, the eigenvalues of $\bm{\Sigma_\Omega}$ are equal to the eigenvalues of $\bm{G}$ repeated $K$ times each. As a result, if $\bm{G}$ is positive definite, so is $\bm{\Sigma_\Omega}$.
	\end{remark}
	
	When the noise process is a stationary Gaussian process, the optimal input is also a stationary Gaussian process \cite{Cover}. It is also optimal to assume the input has zero mean. Then, the optimal input data sequence is a zero mean, stationary Gaussian random process. As both input and noise are Gaussian, the output $\bm{Y}$ is also Gaussian with the covariance matrix $\bm{\Sigma_Y}$, leading to the differential entropy                          
	\begin{eqnarray}
		h(\bm{Y})&=&\log_2\left[\left(2\pi e\right)^{KN}\left(\det(\bm{\Sigma_Y})\right)\right]. 
	\end{eqnarray}
By the definition of covariance, we can write 
	\begin{align}
		\bm{\Sigma_Y}&=\mathbb{E}[(\bm{H}\bm{A}+\bm{\Omega})(\bm{H}\bm{A}+\bm{\Omega})^\dagger] \notag\\ &=(\bm{\tilde{H}}\otimes\bm{G})\bm{\Sigma_A}(\bm{\tilde{H}}\otimes\bm{G})^\dagger+\sigma_0^2(\bm{I}_L\otimes\bm{G}),\label{eqn:covY}
	\end{align}where the $LN \times LN$ covariance matrix $\bm{\Sigma_A}$ is defined as
\begin{equation}
		\bm{\Sigma_A}=\left[\begin{matrix}
			\bm{K}_{\bm{a},11} & \bm{K}_{\bm{a},12}  &\cdots &\bm{K}_{\bm{a},1L} \\
			\bm{K}_{\bm{a},21} & \bm{K}_{\bm{a},22}  &\cdots &\bm{K}_{\bm{a},2L} \\
			\vdots &\vdots &\ddots &\vdots \\
			\bm{K}_{\bm{a},L1} & \bm{K}_{\bm{a},L2} &\cdots &\bm{K}_{\bm{a},LL} \\
		\end{matrix}\right], \label{eqn:covA}
	\end{equation} 
with $\bm{K}_{\bm{a},ij}=\mathbb{E}[\bm{a}_i\bm{a}_j^\dagger], i,j = 1,2,\cdots,L$. Then, the capacity expression in \eqref{eqn:capacity2} becomes equal to  
	\begin{equation}
		C=\lim_{N\rightarrow\infty}\max_{p(\bm{A}),\mathcal{P}}\frac{1}{N}\log_2\frac{\det(\bm{\Sigma_Y})}{\det(\bm{\Sigma_\Omega})}.
	\end{equation}
The derivation of MIMO FTN capacity expression is shown in \eqref{eqn:mimostepbeg}-\eqref{eqn:mimostepend}.   
	\begin{figure*}[t]
	\centering
		\begin{align}
			C&=\lim_{N\rightarrow\infty}\max_{p(\bm{A}), \mathcal{P}}\frac{1}{N}\log_2\frac{\det(\bm{\Sigma_Y})}{\det(\bm{\Sigma_\Omega})} \label{eqn:mimostepbeg}\\
			&=\lim_{N\rightarrow\infty}\max_{\bm{\Sigma_A}, \mathcal{P}}\frac{1}{N}\log_2\det\left[\left(\bm{\Sigma_\Omega}+(\bm{\tilde{H}}\otimes\bm{G})\bm{\Sigma_A}(\bm{\tilde{H}}\otimes\bm{G})^\dagger\right)\bm{\Sigma_\Omega}^{-1}\right]\\
			&\overset{(a)}{=}\lim_{N\rightarrow\infty}\max_{\bm{\Sigma_A}, \mathcal{P}}\frac{1}{N}\log_2\det\left[\bm{I}_{LN}+(\bm{\tilde{H}}\otimes\bm{G})\bm{\Sigma_A}(\bm{\tilde{H}}^\dagger\otimes\bm{I}_N)(\bm{I}_K\otimes\bm{G}^\dagger)\bm{\Sigma_\Omega}^{-1}\right] \label{eqn:inverseSigmaeta}\\
			&\overset{(b)}{=}\lim_{N\rightarrow\infty}\max_{\bm{\Sigma_A}, \mathcal{P}}\frac{1}{N}\log_2\det\left[\bm{I}_{LN}+\bm{\Sigma_A}(\bm{\tilde{H}}^\dagger\otimes\bm{I}_N)(\bm{I}_K\otimes\bm{G}^\dagger)\bm{\Sigma_\Omega}^{-1}(\bm{\tilde{H}}\otimes\bm{G})\right] \\
			&=\lim_{N\rightarrow\infty}\max_{\bm{\Sigma_A}, \mathcal{P}}\frac{1}{N}\log_2\det\left[\bm{I}_{LN}+\bm{\Sigma_A}(\bm{\tilde{H}}^\dagger\otimes\bm{I}_N)(\bm{I}_K\otimes\bm{G}^\dagger)(\bm{I}_K\otimes\bm{G}^{-1})\sigma_0^{-2}(\bm{\tilde{H}}\otimes\bm{G})\right] \\
			&\overset{(c)}{=}\lim_{N\rightarrow\infty}\max_{\bm{\Sigma_A}, \mathcal{P}}\frac{1}{N}\log_2\det\left[\bm{I}_{LN}+\sigma_0^{-2}\bm{\Sigma_A}(\bm{\tilde{H}}^\dagger\otimes\bm{I}_N)(\bm{\tilde{H}}\otimes\bm{G})\right] \\
			&\overset{(d)}{=}\lim_{N\rightarrow\infty}\max_{\bm{\Sigma_A}, \mathcal{P}}\frac{1}{N}\log_2\det\left[\bm{I}_{LN}+\sigma_0^{-2}\bm{\Sigma_A}(\bm{\tilde{H}}^\dagger\otimes\bm{I}_N)(\bm{\tilde{H}}\otimes\bm{I}_N)(\bm{I}_L\otimes\bm{G})\right] \\	
			&\overset{(e)}{=}\lim_{N\rightarrow\infty}\max_{\bm{\Sigma_A}, \mathcal{P}}\frac{1}{N}\log_2\det\left[\bm{I}_{LN}+\sigma_0^{-2}\bm{\Sigma_A}(\bm{\tilde{H}}^\dagger\bm{\tilde{H}}\otimes\bm{I}_N)(\bm{I}_L\otimes\bm{G})\right]\\
			&\overset{(f)}{=}\lim_{N\rightarrow\infty}\max_{\bm{\Sigma_A}, \mathcal{P}}\frac{1}{N}\log_2\det\left[\bm{I}_{LN}+\sigma_0^{-2}(\bm{I}_L\otimes\bm{G})\bm{\Sigma_A}(\bm{\tilde{H}}^\dagger\bm{\tilde{H}}\otimes\bm{I}_N)\right]\label{eqn:mimostepend}
		\end{align}
	\end{figure*}
	The steps (a), (c), (d) and (e) are because of the mixed product property of the Kronecker product, which states that an $A \times A$ matrix $\bm{A}$ and an $B \times B$ matrix $\bm{B}$ has the Kronecker product $(\bm{A}\otimes\bm{B})=(\bm{A}\otimes\bm{I}_B)(\bm{I}_A\otimes\bm{B})$\cite{matrixbook}. In (a), we also make use of the property $(\bm{\tilde{H}}\otimes\bm{G})^\dagger=(\bm{\tilde{H}}^\dagger\otimes\bm{G}^\dagger)$. The steps (b) and (f) are because of Sylvester's determinant identity \cite{matrix}. Moreover, notice that in (c), $(\bm{I}_K\otimes\bm{G}^\dagger)(\bm{I}_K\otimes\bm{G}^{-1})=(\bm{I}_K\bm{I}_K\otimes\bm{G}^\dagger\bm{G}^{-1})=(\bm{I}_K\otimes\bm{I}_N)=\bm{I}_{KN}$. To make sure that the matrix inverse $\bm{\Sigma_\Omega}^{-1}=\sigma_0^{-2}(\bm{I}_K\otimes\bm{G}^{-1})$ in \eqref{eqn:inverseSigmaeta} exists, we require $\bm{G}$ to be invertible. In other words, \eqref{eqn:mimostepend} is valid only if $\bm{G}$ is positive definite or all its eigenvalues are strictly positive.


	\subsection{Power Constraint} \label{subsec3B}
 In FTN signaling, the signal transmitted from the $l$th antenna, $x_l(t)$ has the following autocorrelation function $\hat{R}_{x_l}(t,t-\tau)$ 
\begin{align}
		&\hat{R}_{x_l}(t,t-\tau)\notag\\
		&=\sum_{m=-\infty}^{+\infty}\sum_{n=-\infty}^{+\infty}R_{a_l}[m-n]p(t-m\delta T)p^*(t-\tau-n\delta T),
	\end{align}where $R_{a_l}[n]$ is the autocorrelation function of the data sequence $a_l[n]$. One can observe that $\hat{R}_{x_l}(t,t-\tau)$ is a cyclostationary process. For cyclostationary processes, we need to find the average of the autocorrelation function over one period and evaluate its value at 0 to find the average power.    
	Then, the power of $x_l(t)$ becomes
	\begin{equation}
	P_{x_l}=R_{x_l}(0)=\mathbb{E}[x_l(t)x_l^*(t)]=\frac{1}{\delta T}\sum_{m=-\infty}^{+\infty}R_{a_l}[m]g[-m],
	\end{equation}
	and the total power $P_x$ equals the summation of all the $P_{x_l}$'s, namely,  
	\begin{align}
	P_x=\sum_{l=1}^{L}R_{x_l}(0)=\frac{1}{\delta T}\sum_{l=1}^{L}\sum_{m=-\infty}^{+\infty}R_{a_l}[m]g[-m]. \label{eq43}
	\end{align}
	If we express the total power constraint in a matrix multiplication form, we observe that the following limit is also equal to $P_x$, as
	\begin{align}
		&\lim_{N\rightarrow\infty}\frac{1}{N\delta T}\text{tr}\left(\left[\begin{matrix}
			\bm{G} & 0  &\cdots &0 \\
			0 & \bm{G} &\cdots &0 \\
			\vdots & \vdots&\ddots &\vdots \\
			0 & 0  &\cdots &\bm{G} \\
		\end{matrix}\right]\bm{\Sigma_A}\right) \label{eqn:mimoconstr} \allowdisplaybreaks\\
		&=\lim_{N\rightarrow\infty}\frac{1}{N\delta T}\text{tr}\left((\bm{I}_L\otimes\bm{G})\bm{\Sigma_A}\right) \allowdisplaybreaks\\
		&= \lim_{N\rightarrow\infty}\frac{1}{\delta T}\sum_{l=0}^{L-1}\sum_{m=-(N-1)}^{N-1}R_{a_l}^*[m]g^*[-m]\frac{N-|m|}{N} \allowdisplaybreaks\\
		&= \frac{1}{\delta T}\sum_{l=0}^{L-1}\sum_{m=-\infty}^{+\infty}(R_{a_l}[m]g[-m])^* \\
		&=P_x.
	\end{align}In orthogonal transmission schemes, the input power constraint is simply a constraint on the trace of $\bm{\Sigma_A}$. However, in FTN the intentional ISI generated at the transmitter changes the power constraint and we observe that the power constraint is now on the trace of the product of $\bm{\Sigma_A}$ and $(\bm{I}_L\otimes\bm{G})$. 
Finally, the power constraint $\mathcal{P}$ in \eqref{eqn:capacity} becomes $P_x \leq P$. 

	\subsection{Capacity Calculation}
	\label{subsec3C}


	Combining \eqref{eqn:mimostepend} and \eqref{eqn:mimoconstr}, the capacity optimization problem for MIMO FTN becomes 
	\begin{subequations}\label{eqn:optMIMO}
		\begin{align}
			&C=\lim_{N\rightarrow\infty}\max_{\bm{\Sigma_A}}\frac{1}{N}\log_2\det\bigg(\bm{I}_{LN}+\sigma_0^{-2}(\bm{I}_L\otimes\bm{G})\bm{\Sigma_A} (\bm{\tilde{H}}^\dagger\bm{\tilde{H}}\otimes\bm{I}_N)\bigg) \label{eqn:FTNMIMOobj}\\
			&s.t.\quad\lim_{N\rightarrow\infty}\frac{1}{N\delta T}\text{tr}\left((\bm{I}_L\otimes\bm{G})\bm{\Sigma_A}\right) \leq P.\label{eqn:ftnmimopower}
		\end{align}
	\end{subequations}
	Define 
	\begin{eqnarray}
		\bm{K} & \triangleq & (\bm{I}_L\otimes\bm{G})\bm{\Sigma_A} \label{eqn:defk} \\
		\bm{H}_{eq} &\triangleq & (\bm{\tilde{H}}\otimes\bm{I}_N) \label{eqn:Heq}\\
		\bm{W}  & \triangleq & \bm{H}_{eq}^\dagger\bm{H}_{eq} \quad \notag\\
		&= & (\bm{\tilde{H}}\otimes\bm{I}_N)^\dagger(\bm{\tilde{H}}\otimes\bm{I}_N) \quad \notag\\
		&=& (\bm{\tilde{H}}^\dagger\bm{\tilde{H}})\otimes\bm{I}_N. \label{eqn:W}
	\end{eqnarray}
	Then the problem in \eqref{eqn:optMIMO} can be written as an optimization problem about the matrix $\bm{K}$ for the extended equivalent channel $\bm{H}_{eq}$ as
	\begin{subequations}\label{eqn:mimoFTN}
		\begin{align}	&C=\lim_{N\rightarrow\infty}C_N=\lim_{N\rightarrow\infty}\max_{\bm{K}}\frac{1}{N}\log_2\det\bigg(\bm{I}_{LN}+\sigma_0^{-2}\bm{K}\bm{W}\bigg) \label{eqn:mimocap} \\
			&s.t.\quad\lim_{N\rightarrow\infty}\frac{1}{N\delta T}\text{tr}(\bm{K}) \leq P.  \label{eqn:mimoconst}
		\end{align}
	\end{subequations}
	In order to solve this optimization problem, we should solve it for each $N$, and then take the limit as $N \rightarrow \infty$. For a specific $N$, the above optimization is equivalent to the well-known waterfilling problem for a regular MIMO channel with the channel gain matrix $\bm{H}_{eq}$ \cite{telatar}. 
	Thus, we will first solve for \eqref{eqn:mimoFTN} for a fixed $N$ to obtain $C_N$ in \eqref{eqn:mimocap} and then take $N$ to infinity. If $C_N$ converges to a constant, the limit will then be the capacity $C$.
	
	To find the solution for a fixed $N$, we write the eigen-decomposition for the composite covariance matrix $\bm{K}$ as  $\bm{K}=\bm{P_K}\bm{\Lambda_K}\bm{P_K}^{-1}$, where $\bm{P_K}$ is composed of the eigenvectors of $\bm{K}$ and $\bm{\Lambda_K}$ is a diagonal matrix with the eigenvalues of $\bm{K}$, $\alpha_1, \alpha_2, \cdots, \alpha_{NL}$, on the main diagonal. We also perform eigen-decomposition on $\bm{W}$ and write $\bm{W}=\bm{V_W\Lambda_WV_W}^\dagger$. Because $\bm{W}$ is Hermitian,  
	$\bm{V_W}$ is a unitary matrix and it is composed of the eigenvectors of $\bm{W}$. Moreover, $\bm{\Lambda_W}$ is a diagonal matrix with the eigenvalues of $\bm{W}$ on the main diagonal. We define
	\begin{equation}\bm{\tilde{H}}^\dagger\bm{\tilde{H}} \triangleq \bm{Z}=\bm{V_Z}\bm{\Lambda_Z}\bm{V_Z}^\dagger, \label{eqn:VZ} \end{equation} where the unitary matrix $\bm{V_Z}$ is composed of the eigenvectors of $\bm{Z}$. 
	Then, for a fixed $N$, we can write and bound the determinant in \eqref{eqn:mimocap} using Hadamard's inequality as \cite{Cover}
	\begin{align}
		C_N&=\max_{\bm{K}}\frac{1}{N}\log_2\det\bigg(\bm{I}_{LN}+\sigma_0^{-2}\bm{P_K}\bm{\Lambda_K}\bm{P_K}^{-1}\bm{V_W}\bm{\Lambda_W}\bm{V_W}^\dagger\bigg) \\
		&\leq\max_{\bm{\Lambda_K}}\frac{1}{N}\log_2\det\bigg(\bm{I}_{LN}+\sigma_0^{-2}\bm{\Lambda_K}\bm{\Lambda_W}\bigg).
	\end{align}
	The above inequality is achieved with equality when $\bm{KW}$ becomes equal to the diagonal matrix $\bm{\Lambda_K}\bm{\Lambda_W}$. Thus, we choose $\bm{P_K}=\bm{V_W}$. Note that, as $\bm{V_W}$ is unitary, this choice also makes the matrix $\bm{K}$ Hermitian.  

	Assuming the product matrix $\bm{\tilde{H}}^\dagger\bm{\tilde{H}}$  has eigenvalues $\tau_1,\cdots, \tau_L$, we can write the eigenvalues of 
	$\bm{W}$ in \eqref{eqn:W} due to Lemma~\ref{lem:kroneig}. As the identity matrix has all of its eigenvalues equal to 1, the eigenvalues of $\bm{W}$ repeat themselves $N$ times and we write $\text{eig}(\bm{W})_{j*N+i}=\tau_i$ for all $j=0,1,\cdots, L-1,  i=1,2,\cdots, N$. Then, we can write the optimization problem in the form of a summation of logarithms as 
	\begin{subequations}
		\begin{align}
			&C_N=\max_{\alpha_1, \alpha_2,\cdots,\alpha_{NL}}\frac{1}{N}\sum_{i=1}^{NL}\log_2\left(1+\sigma_0^{-2}\alpha_i\tau_{\lceil i/N\rceil}\right) \label{eqn:mimoftn2}\\
			&s.t.\quad\lim_{N\rightarrow\infty}\frac{1}{N\delta T}\sum_{i=1}^{NL}\alpha_i \leq P,
		\end{align}
	\end{subequations}
	where $\lceil i/N\rceil$ takes the ceiling of fraction $i/N$. 
	The solution to this problem is obtained via waterfilling as 
	\begin{align}
		\bar{\alpha}_i&=\sigma_0^2\left(\frac{1}{\mu}-\frac{1}{\tau_{\lceil i/N\rceil}}\right)^+, \label{eqn:optalpha}
	\end{align}
	where $i =1,2,\cdots,NL$, and $(a)^+=\max(a,0)$. Then, $\mu$ is solved by \begin{eqnarray}\sigma_0^2\sum_{i=1}^{NL}\left(\frac{1}{\mu}-\frac{1}{\tau_{\lceil i/N\rceil}}\right)^+=NP\delta T, \label{eqn:mimoftnpower}
	\end{eqnarray} and the optimal covariance matrix becomes \begin{equation}\bar{\bm{K}}=\bm{V_W}\bm{\bar{\Lambda}_K}\bm{V_W}^\dagger, \label{eqn:optKK}\end{equation}where $\bm{\bar{\Lambda}_K}=\text{diag}\{\bar{\alpha}_1,\bar{\alpha}_2,\cdots,\bar{\alpha}_{NL}\}$. 
	Using this optimal covariance matrix $\bar{\bm{K}}$ in \eqref{eqn:optKK} and \eqref{eqn:defk}, we can find the optimal input covariance matrix $\bm{\bar{\Sigma}_A}$ as
	\begin{equation}
		\bm{\bar{\Sigma}_A}=(\bm{I}_L\otimes\bm{G})^{-1}\bar{\bm{K}}. \label{eqn:solumimo}
	\end{equation}
At this stage, we have to guarantee that  $\bm{\bar{\Sigma}_A}$ is a valid covariance matrix, so it has to be both positive semidefinite and Hermitian. 
	\begin{lemma}\label{lem:solumimo}
		The matrix $\bm{\bar{\Sigma}_A}$ is always Hermitian and is positive definite. 
	\end{lemma}
	\begin{proof}
		The proof is provided in Appendix \ref{appE}.
	\end{proof}

	Overall, combining \eqref{eqn:mimoftn2}, \eqref{eqn:optalpha}, \eqref{eqn:solumimo} and Lemma~\ref{lem:solumimo}, we write that the maximum mutual information for a specific $N$ is equal to 
	\begin{equation}
		C_N=\sum_{i=1}^{L}\log_2(1+\sigma_0^{-2}\bar{\alpha}_{(i-1)N+1}\tau_{i}),\quad \text{when } \delta(1+\beta)\geq1.
		\label{eq56}
	\end{equation} 
	As \eqref{eq56} does not depend on $N$, we conclude that this solution applies to any $N$ and the MIMO FTN capacity is given in Theorem~\ref{thm:MIMOFTNC}.
	{\begin{remark}
	Note that Lemma~\ref{lem:solumimo} is valid for all $\delta$. However, due to Lemma~\ref{lem:Gposdef}, to provide numerical stability in the calculation of $\bm{G}^{-1}$, we need $\delta(1+\beta)\geq1$ in \eqref{eq56}.
	\end{remark}}
   \begin{theorem}\label{thm:MIMOFTNC}
		When a raised cosine pulse with roll-off factor $\beta$ is used as $g(t)$, the capacity of MIMO FTN signaling with acceleration factor $\delta$ becomes \begin{equation}
			C=\sum_{i=1}^{L}\log_2(1+\sigma_0^{-2}\bar{\alpha}_{(i-1)N+1}\tau_{i}),\quad \text{when } {\delta(1+\beta)\geq1} \label{eqn:capmimo}
		\end{equation}in bits/channel use. As we transmit $\frac{W}{\delta}$ symbols per second and the bandwidth is $(1+\beta)W$ Hz, we can rewrite \eqref{eqn:capmimo} to obtain
		\begin{equation}
		\bar{C}=\frac{1}{\delta(1+\beta)}\sum_{i=1}^{L}\log_2(1+\sigma_0^{-2}\bar{\alpha}_{(i-1)N+1}\tau_{i}), \quad  \text{when } {\delta(1+\beta)\geq1} \label{eqn:capmimobitshz}
		\end{equation}in bit/s/Hz.
	\end{theorem}	\begin{remark}	\label{rem6}
	  In \eqref{eqn:capmimo}, $\bar{\alpha}_i$ of \eqref{eqn:optalpha} depends on $\delta$, because of the power constraint in \eqref{eqn:mimoftnpower}. In other words, the acceleration factor is implicitly included in the result. 
	\end{remark}
	
\begin{remark}
	The optimal transmission scheme is equivalent to precoding in time and waterfilling in space. Precoding in time results in perfect ISI mitigation at the transmitter side.
	\end{remark}

	This remark is justified as follows. First, note that in \eqref{eqn:optalpha} the value for $\tau_{\lceil i/N \rceil}$ is the same for all groups of $i =\{1,\cdots,N\}$, $i = \{N+1,\cdots, 2N\}$, until $i= \{(N-1)L+1,\cdots,NL\}$, thus there are only $L$ distinct $\alpha_i$ values. The same $\bar{\alpha}_i$ values repeat themselves $N$ times consecutively, with $L$ groups in total in $\bm{\bar{\Lambda}_K}$ in \eqref{eqn:optKK}, and we can rewrite
	$\bm{\bar{\Lambda}_K}$ as $\bm{\bar{\Lambda}_K}=\bm{\Lambda}_{\bar{\alpha}}\otimes\bm{I}_N$, where $\bm{\Lambda}_{\bar{\alpha}}=\text{diag}\{\bar{\alpha}_1,\bar{\alpha}_{N+1}, \cdots, \bar{\alpha}_{N(L-1)+1}\}$ is the diagonal matrix composed of the distinct $\bar{\alpha}_i$ values.

As the details are shown in Appendix \ref{appE} in \eqref{eqn:eigoptsolu}, we can write  the optimal covariance matrix $\bm{\bar{\Sigma}_A}$ in  \eqref{eqn:solumimo} as
	\begin{align}
	    \bm{\bar{\Sigma}_A}&=(\bm{V_Z}\otimes\bm{V_G})(\bm{I}_L\otimes\bm{\Lambda_G}^{-1})\bm{\bar{\Lambda}_K}(\bm{V}_{\bm{Z}}\otimes\bm{V_G})^\dagger,
	\end{align} where
	 $\bm{V_G}\bm{\Lambda_G}^{-1}\bm{V_G}^\dagger$ is defined as the eigenvalue decomposition for $\bm{G}^{-1}$. Then, inserting $\bm{\bar{\Lambda}_K}=\bm{\Lambda}_{\bar{\alpha}}\otimes\bm{I}_N$ in the above equation, we obtain
	\begin{align}
	    \bm{\bar{\Sigma}_A}&=(\bm{V_Z}\otimes\bm{V_G})(\bm{I}_L\otimes\bm{\Lambda_G}^{-1})(\bm{\Lambda}_{\bar{\alpha}}\otimes\bm{I}_N)(\bm{V}_{\bm{Z}}\otimes\bm{V_G})^\dagger\\
	    &=(\bm{V_Z}\otimes\bm{V_G})(\bm{\Lambda}_{\bar{\alpha}}\otimes\bm{\Lambda_G}^{-1})(\bm{V}_{\bm{Z}}\otimes\bm{V_G})^\dagger \label{eqn:ll}\\
	    &=(\bm{V_Z}\bm{\Lambda}_{\bar{\alpha}}\bm{V_Z}^\dagger)\otimes\bm{G}^{-1}, \label{eqn:lll}
	\end{align}
as a valid covariance matrix due to Lemma~\ref{lem:solumimo}. Next, we insert the optimal input covariance matrix \eqref{eqn:lll} into the MIMO FTN capacity expression  \eqref{eqn:FTNMIMOobj} to obtain \eqref{eqn:begin}. Using the properties for the Kronecker product, we can simplify \eqref{eqn:begin} as in \eqref{eqn:chnlinverse}-\eqref{eqn:kkk} to obtain the final form for the MIMO FTN capacity expression as in \eqref{eqn:end}.
	\begin{figure*}[t]
		\centering
			\begin{align}
	    C&=\lim_{N\rightarrow\infty}\frac{1}{N}\log_2\det\left[\bm{I}_{LN}+\sigma_0^{-2}(\bm{I}_L\otimes\bm{G}) \left(\left(\bm{V_Z}\bm{\Lambda}_{\bar{\alpha}}\bm{V_Z}^\dagger\right)\otimes\bm{G}^{-1} \right) (\bm{\tilde{H}}^\dagger\bm{\tilde{H}}\otimes\bm{I}_N)\right] \label{eqn:begin}\\
		&=\lim_{N\rightarrow\infty}\frac{1}{N}\log_2\det\left[\bm{I}_{LN}+\sigma_0^{-2}\left(  \left(\bm{V_Z}\bm{\Lambda}_{\bar{\alpha}}\bm{V_Z}^\dagger\right)  \otimes  \bm{I}_N\right) \left(\bm{\tilde{H}}^\dagger\bm{\tilde{H}}\otimes\bm{I}_N\right)\right] \label{eqn:chnlinverse}\\
		&=\lim_{N\rightarrow\infty}\frac{1}{N}\log_2\det\left[\bm{I}_{L}\otimes\bm{I}_N+\sigma_0^{-2}\left( \left (\bm{V_Z}\bm{\Lambda}_{\bar{\alpha}}\bm{V_Z}^\dagger\right) \bm{Z} \otimes  \bm{I}_N\right)\right] \\
		&=\lim_{N\rightarrow\infty}\frac{1}{N}\log_2\det\left[\left( \bm{I}_{L}+\sigma_0^{-2} \left(\bm{V_Z}\bm{\Lambda}_{\bar{\alpha}}\bm{V_Z}^\dagger\right) \bm{Z} \right) \otimes  \bm{I}_N\right] \\
		&=\lim_{N\rightarrow\infty}\frac{1}{N}N\log_2\det\left[ \bm{I}_{L}+\sigma_0^{-2} \left(\bm{V_Z}\bm{\Lambda}_{\bar{\alpha}}\bm{\Lambda_Z}\bm{V_Z}^\dagger\right) \right] \label{eqn:kkk}\\
		&=\log_2\det\left( \bm{I}_{L}+\sigma_0^{-2} \bm{\Lambda}_{\bar{\alpha}}\bm{\Lambda_Z} \right) \label{eqn:end}
	\end{align}
	\end{figure*}
	
	This result in \eqref{eqn:end} clearly shows that the best achievable scheme is waterfilling for the MIMO channel. The best  $\bar{\alpha}_i$'s, are solely determined by the MIMO channel $\bm{\tilde{H}}$ in \eqref{eqn:optalpha}. Comparing \eqref{eqn:begin} with \eqref{eqn:chnlinverse}, we also observe that the effect of ISI is completely mitigated, the original input data $\bm{A}$ is \emph{precoded} for ISI removal via $\bm{G}^{-1}$ in \eqref{eqn:lll}, and the matrix $\bm{G}$ no longer appears in \eqref{eqn:chnlinverse}-\eqref{eqn:end}. 
	We also express that in MIMO FTN the symbols transmitted from each antenna have the same autocorrelation in time but only differ in their power levels, determined according to the MIMO channel $\bm{\tilde{H}}$. 
{
\begin{remark}
{In addition to the numerical instability of $\bm{G}^{-1}$ mentioned in Lemma \ref{lem:Gposdef}, there will be other hurdles when implementing FTN in practice. For example, practically ideal root raised cosine pulses are not possible, and truncation in time is performed. 
This truncation generally broadens the spectrum \cite{spectrumbroaden} more than it does in orthogonal signaling. Therefore, overall spectral efficiency calculations in FTN should take spectrum broadening into account as well.}
\end{remark}

Finally, note that a complete analysis which accounts for $\delta(1+\beta)< 1$, is a very complex problem. In \cite{R3asymp}, the authors prove that binary signaling is asymptotically optimal and achieves the capacity of FTN signaling with Gaussian inputs. In addition \cite{ganji}, claims a constant capacity for $\delta(1+\beta)< 1$, which is equal to the capacity achieved for $\delta = \frac{1}{1+\beta}$ for SISO FTN. This paper is a work in progress, and this issue of $\delta(1+\beta)< 1$ is one of our top priorities for future work. 
}


    \section{Capacity of MIMO FTN Signaling in Frequency Selective Channels}\label{sec:freqsel}

FS fading occurs when the channel spectrum is not constant over the transmission bandwidth \cite{goldsmith}. Frequency selectivity distorts the transmitted signal. It presents itself in time domain as multipath. Signal arrives at the receiver from different paths, each with a different delay and gain. This effect can be modeled by a tapped-delay-filter with $J$ taps. In this section we solve the capacity of MIMO FTN signaling in FS channels. 
	
Similar to the SISO model in {\cite{eigendecomposition}}, in MIMO, the signal component $x_{kl}(t)$ at the output of the matched filter at the $k$th receive antenna, due to the $l$th transmit antenna has the expression 
	 	\begin{equation}
		x_{kl}(t)=\sum_{n=0}^{N-1}\sum_{j=0}^{J-1}h^j_{kl}a_l[n]g(t-(j+n)\delta T), \label{eqn:fsfctimemodel}
	\end{equation}
	where $h^j_{kl}$ is the channel gain from the $l$th transmit antenna to the $k$th receive antenna for the $j$th tap. Note that multiple taps result in ISI, in addition to the ISI due to FTN signaling. We would like to express that, every channel from one transmit antenna to a receive antenna does not need to have equal number of taps. The number of taps $J$ is simply the maximum number of taps in all channels. For those channels with the number of taps less than $J$, we can assume the remaining channel coefficients are zero, and every channel has $J$ taps. If we sample this signal at $\delta T$ intervals, as in {\cite{eigendecomposition}}, the resulting $N$ samples can be written in vector form as  {\eqref{eqn:tildeGkl1}.}
	\begin{figure*}
	    \centering
    \begin{align}
		\bm{x}_{kl}&=
			\left[\begin{matrix}
				x_{kl}[0] \\	x_{kl}[1]\\ \vdots \\ 	x_{kl}[N-1] 
			\end{matrix}\right] \allowdisplaybreaks   
		 = \left[\begin{matrix}
			\sum_{j=0}^{J-1}h^j_{kl}g[-j] & \sum_{j=0}^{J-1}h^j_{kl}g[-1-j] & \cdots & \sum_{j=0}^{J-1}h^j_{kl}g[1-N-j] \\	\sum_{j=0}^{J-1}h^j_{kl}g[1-j] & \sum_{j=0}^{J-1}h^j_{kl}g[-j] & \cdots & \sum_{j=0}^{J-1}h^j_{kl}g[2-N-j]   \\ \vdots& \vdots &\ddots & \vdots \label{eqn:tildeGkl1}\\ 	\sum_{j=0}^{J-1}h^j_{kl}g[N-1-j] & \sum_{j=0}^{J-1}h^j_{kl}g[N-2-j] & \cdots & \sum_{j=0}^{J-1}h^j_{kl}g[-j] 
		\end{matrix}\right] \left[\begin{matrix}
		a_l[0] \\ a_l[1] \\ \vdots \\ a_l[N-1]
	\end{matrix}\right]
	\end{align}
	\end{figure*}
	We can also write $\bm{x}_{kl}$ as
	\begin{equation} 
	 	\bm{x}_{kl}=\tilde{\bm{G}}_{kl}\bm{a}_l. \label{eqn:tildeGkl}
	\end{equation}
	
	When we define $\bm{G}^j$, which is a $j$-\emph{shifted} version of $\bm{G}$, as
	\begin{equation}
		\bm{G}^j \triangleq \left[\begin{matrix}
			g[-j] & g[-1-j] & \cdots & g[1-N-j] \\	g[1-j] & g[-j] & \cdots & g[2-N-j]   \\ \vdots& \vdots &\ddots & \vdots \\ 	g[N-1-j] & g[N-2-j] & \cdots & g[-j] 
		\end{matrix}\right],\label{eqn:defGj}
	\end{equation} 
	with $\bm{G}^0=\bm{G}$, we have $\tilde{\bm{G}}_{kl}$ in \eqref{eqn:tildeGkl} as
	\begin{equation}
		\tilde{\bm{G}}_{kl}=\sum_{j=0}^{J-1}h^j_{kl}\bm{G}^j.
	\end{equation}
	Thus, the samples at the output of the matched filters at all receive antennas become
	\begin{align}
		\bm{Y}=&\left[\begin{matrix}
			\tilde{\bm{G}}_{11} & \tilde{\bm{G}}_{12} & \cdots & \tilde{\bm{G}}_{1L} \\	\tilde{\bm{G}}_{21} & \tilde{\bm{G}}_{22} & \cdots & \tilde{\bm{G}}_{2L}   \\ \vdots& \vdots &\ddots & \vdots \\ 	\tilde{\bm{G}}_{K1} & \tilde{\bm{G}}_{K2} & \cdots & \tilde{\bm{G}}_{KL} 
		\end{matrix}\right] \left[\begin{matrix}
		\bm{a}_1 \\ \bm{a}_2 \\ \vdots \\ \bm{a}_L 
	\end{matrix}\right] + \bm{\Omega} \\
	 =& \left[\sum_{j=0}^{J-1}(\tilde{\bm{H}}^j\otimes\bm{G}^j)\right]\bm{A} +\bm{\Omega}, \label{eqn:71} \\
	 \triangleq & \bm{G_H}\bm{A} +\bm{\Omega} \label{eqn:72}
	\end{align}
	where $(\tilde{\bm{H}}^j)_{k,l}=h^j_{kl}$.  
	As $\sum_{j=0}^{J-1}(\tilde{\bm{H}}^j\otimes\bm{G}^j)=\bm{G_H}$ in \eqref{eqn:72}, we derive the capacity expression in FS fading channels as in {\eqref{eqn:73}-\eqref{eqn:73_3}}, 
	\begin{figure*}[b]
	    \centering
\begin{align}
		C^{\text{FS}}&=\lim_{N\rightarrow\infty}\max_{\bm{\Sigma_A},\mathcal{P}}\frac{1}{N}\log_2\det\left[\bm{I}_{KN}+\sigma^{-2}_0\bm{G_H}\bm{\Sigma_A}\bm{G_H}^\dagger(\bm{I}_K\otimes\bm{G}^{-1})\right] \label{eqn:73}\\
		&= \lim_{N\rightarrow\infty}\max_{\bm{\Sigma_A},\mathcal{P}}\frac{1}{N}\log_2\det\left[\bm{I}_{LN}+\sigma^{-2}_0\bm{\Sigma_A}\bm{G_H}^\dagger(\bm{I}_K\otimes\bm{G}^{-\frac{1}{2}})(\bm{I}_K\otimes\bm{G}^{-\frac{1}{2}})\bm{G_H}\right] \\
		&= \lim_{N\rightarrow\infty}\max_{\bm{\Sigma_A},\mathcal{P}}\frac{1}{N}\log_2\det\left[\bm{I}_{LN}+\sigma^{-2}_0\bm{\Sigma_A}\bm{\Phi}^\dagger\bm{\Phi}\right] \label{eqn:73_3}
	\end{align}
	\end{figure*}
	where $\bm{\Phi}\triangleq (\bm{I}_K\otimes\bm{G}^{-\frac{1}{2}})\bm{G_H}$. We also assume that $\bm{G}$ is invertible in \eqref{eqn:73}. Because $\bm{\Phi}^\dagger\bm{\Phi}$ is a Hermitian matrix, we can write its eigenvalue decomposition as $\bm{\Phi}^\dagger\bm{\Phi}=\bm{U_\Phi}\bm{\Lambda_\Phi}\bm{U_\Phi}^\dagger$. Following the same approach as in {\cite{eigendecomposition}}, we first solve the optimization problem for a fixed $N$ and then send $N$ to infinity, resulting in the capacity if the limit exists. Applying Hadamard's inequality \cite{Cover} to get an upper bound on $C^{\text{FS}}_N$, we write
	\begin{align}
		C^{\text{FS}}_N&\leq \frac{1}{N}\log_2\det\bigg(\bm{I}_{LN}+\sigma^{-2}_0\bm{\Lambda_\Sigma}\bm{\Lambda_\Phi}\bigg), \label{eqn:Hadamard} \\
		&= \frac{1}{N}\sum_{i=0}^{LN-1}\log_2(1+\sigma^{-2}_0{\lambda}_i{\phi}_i).
	\end{align}
	where the upper bound is achieved when $\bm{\Sigma_A}$ can be decomposed into $\bm{U_\Phi}\bm{\Lambda_\Sigma}\bm{U_\Phi}^\dagger$, ${\lambda}_i$ and ${\phi}_i$ are eigenvalues of $\bm{\Sigma_A}$ and $\bm{\Phi}$ respectively.
	
	The power constraint for FS channels remains the same as in \eqref{eqn:ftnmimopower}. When \eqref{eqn:Hadamard} is achieved with equality, we can manipulate the power constraint as follows
	\begin{align}
		\lim_{N\rightarrow\infty}\frac{1}{N\delta T}\text{tr}\big((\bm{I}_L\otimes\bm{G})\bm{\Sigma_A}\big)&=\lim_{N\rightarrow\infty}\frac{1}{N\delta T}\text{tr}\big(\bm{U_\Phi}^\dagger(\bm{I}_L\otimes\bm{G})\bm{U_\Phi}\bm{\Lambda_\Sigma}\big) \label{eqn:fsfcdetpowconst}\\
		&=\lim_{N\rightarrow\infty}\frac{1}{N\delta T}\text{tr}(\bm{\Psi}\bm{\Lambda_\Sigma})\\ &=\lim_{N\rightarrow\infty}\frac{1}{N\delta T}\sum_{i=0}^{LN-1}{\psi}_i{\lambda}_i \leq P, \label{eqn:fsfcpowconst}
	\end{align}
	where  $\bm{\Psi} \triangleq \bm{U_\Phi}^\dagger(\bm{I}_L\otimes\bm{G})\bm{U_\Phi}$, and  $\psi_i$'s are the diagonal entries of matrix the $\bm{\Psi}$. Note that $\bm{\Psi}$ is not necessarily a diagonal matrix. The optimization problem then becomes 
	\begin{align}
		&C^{\text{FS}}_N=\max_{{\lambda}_i}\frac{1}{N}\sum_{i=0}^{LN-1}\log_2(1+\sigma^{-2}_0{\lambda}_i{\phi}_i) \\
		&s.t. ~~\frac{1}{N\delta T}\sum_{i=0}^{LN-1}{\psi}_i{\lambda}_i \leq P.
	\end{align}
	Similar to the SISO case in {\cite{eigendecomposition}}, we solve this problem by the Lagrange multiplier method, forming the Lagrangian function and by using Kuhn-Tucker conditions to find the optimal ${\lambda}^*_i$'s as
	\begin{equation}
		{\lambda}^*_i=\max\left(\frac{\delta T\ln2}{\mu {\psi}_i}-\frac{\sigma_0^2}{{\phi}_i},0\right). 
	\end{equation}
As a result, assuming the limit exists, the capacity for FTN signaling in FS MIMO channels become
	\begin{equation}
		C^{\text{FS}} = \lim_{N\rightarrow\infty}\frac{1}{N}\sum_{i=0}^{LN-1}\log_2(1+\sigma^{-2}_0{\lambda}^*_i{\phi}_i).
	\end{equation}
	{The parameters $\phi_i$'s and $\psi_i$'s are determined by the tap coefficients of the MIMO channel as well as the $\bm{G}$ matrix. In other words, all factors in space and frequency have an effect on the optimal transmission power level.   In FS MIMO FTN, in addition to the effect of multiple antennas, both the frequency spectrum of the channel and the frequency spectrum due to FTN ISI, have an influence on the transmitted signal. The optimal transmission scheme includes joint waterfilling in spatial and frequency domains, this is illustrated in more details in the following section.}

{
\section{Analysis in Frequency Domain}{\label{sec:freqdomanalysis}}

We now analyse the above channel capacities in the frequency domain. In the following discussion we still assume $\delta(1+\beta) \geq 1$.

\subsection{Spectrum Analysis of FS MIMO Channels with FTN}\label{subsec:freqdomA}

{We start this subsection with the following lemma.
\begin{lemma}\label{lem:blktoep}
    \cite[Theorem 3]{asynnoma} The generalized Szeg\"o's theorem states that
\begin{equation}
    \underset{N\rightarrow\infty}{\lim}\frac{1}{N}\sum_{l=1}^{NL}F[\lambda_l(\bm{R})] = \int_{-\frac{1}{2}}^{\frac{1}{2}}\sum_{j=1}^{L} F[\lambda_j(\bm{R}(f_n))]df_n, \label{eqn:genszego}
\end{equation}
where $F[.]$ is any function continuous on the range of $f_n$, $\lambda_l(\mathbf{R})$ is the $l$th eigenvalue of block Toeplitz matrix $\bm{R}$ and $\lambda_j(\bm{R}(f_n))$ is the $j$th eigenvalue of matrix $\bm{R}(f_n)$. The matrix $\bm{R}$ is of size $NL \times NL$ and has the form 
\begin{equation}
    \bm{R}=\left[\begin{matrix} \bm{R}_{11} & \bm{R}_{12} & \dots& \bm{R}_{1L} \\
    \bm{R}_{21} & \bm{R}_{22} & \dots& \bm{R}_{2L} \\
    \vdots & \vdots & \ddots& \vdots \\
    \bm{R}_{L1} & \bm{R}_{L2} & \dots & \bm{R}_{LL}
    \end{matrix}\right],
\end{equation}
where $\bm{R}_{ij}$'s are  $N\times N$ Toeplitz matrices individually. The matrix $\bm{R}(f_n)$ is called the generating matrix of $\bm{R}$, since it is composed of the generating functions of $\bm{R}_{ij}$'s, namely, $(\bm{R}(f_n))_{ij}=\mathcal{G}_{ij}(f_n)$, where $\mathcal{G}_{ij}(f_n)$ is the generating function of $\bm{R}_{ij}$, $i,j=1,\dots,L$. {For an arbitrary Toeplitz matrix $\bm{T}_N$, with $(\bm{T}_N)_{ij}=t_{i-j}$, 
its generating function is defined as \cite[Theorem 3]{asynnoma} 
\begin{eqnarray}
\mathcal{G}(\bm{T}_N)=\sum_{k=-\infty}^{\infty}t_ke^{jkf_n}, \quad f_n\in\left[-\frac{1}{2},\frac{1}{2}\right].
\end{eqnarray}} 
\end{lemma}}
In order to express the mutual information in \eqref{eqn:73} in the frequency domain, we first need to verify that Lemma \ref{lem:blktoep}
can be applied, in other words, the matrix inside the determinant in \eqref{eqn:73} is a block Toeplitz matrix. First of all, it is easy to see that $\bm{G_H}$, $(\bm{I}_K\otimes\bm{G}^{-1})$ and $\bm{\Sigma_A}$ are  block Toeplitz matrices when $N\rightarrow\infty$\footnote{According to \cite{gray}, since the inverse of a Toeplitz matrix is asymptotically Toeplitz, $\bm{I}_K\otimes\bm{G}^{-1}$ is also an asymptotically block Toeplitz  matrix.}. 
{ According to \cite{blocktoep}, we have the following lemma on the product of block Toeplitz matrices}
\begin{lemma}{\cite[Theorem 2]{blocktoep}} \label{lem:blktoepprod}
The product of  block Toeplitz matrices is asymptotically a block Toeplitz matrix. 
\end{lemma} 

After manipulating {\eqref{eqn:73}} we can obtain 
\begin{equation}
    C^{FS}=\underset{N\rightarrow\infty}{\lim}\underset{\bm{\Sigma_A,\mathcal{P}}}{\max}\frac{1}{N}\log_2\det\bigg[\bm{I}_{LN}+\sigma^{-2}_0\bm{\Sigma_A}\bm{G_H}^\dagger(\bm{I}_K\otimes\bm{G}^{-1})\bm{G_H}\bigg].  \label{eqn:cfsmanipulated}
\end{equation}
In this equation, the matrix $\bm{G_H}^\dagger(\bm{I}_K\otimes\bm{G}^{-1})\bm{G_H}$ can be written as {\eqref{eqn:89}-\eqref{eqn:92},}
\begin{figure*}[ht]
    \centering
\begin{align}
    &\bm{G_H}^\dagger(\bm{I}_K\otimes\bm{G}^{-1})\bm{G_H} =\left(\sum_{j=0}^{J-1}\left(\tilde{\bm{H}}^{j\dagger}\otimes\bm{G}^{j\dagger}\bm{G}^{-1}\right)\right)\left(\sum_{i=0}^{J-1}\left(\tilde{\bm{H}}^{i}\otimes\bm{G}^{i}\right)\right)=\sum_{j=0}^{J-1}\sum_{i=0}^{J-1}\left(\tilde{\bm{H}}^{j\dagger}\tilde{\bm{H}}^{i}\otimes\bm{G}^{j\dagger}\bm{G}^{-1}\bm{G}^{i}\right) \label{eqn:89}\\
    &=\left[\begin{matrix}
    \sum_{i,j}(\sum_{k=1}^{K}h^{j*}_{k1}h^{i}_{k1})\bm{G}^{j\dagger}\bm{G}^{-1}\bm{G}^{i} & \sum_{i,j}(\sum_{k=1}^{K}h^{j*}_{k1}h^{i}_{k2})\bm{G}^{j\dagger}\bm{G}^{-1}\bm{G}^{i} & \dots & \sum_{i,j}(\sum_{k=1}^{K}h^{j*}_{k1}h^{i}_{kL})\bm{G}^{j\dagger}\bm{G}^{-1}\bm{G}^{i}   \\
    \sum_{i,j}(\sum_{k=1}^{K}h^{j*}_{k2}h^{i}_{k1})\bm{G}^{j\dagger}\bm{G}^{-1}\bm{G}^{i} & \sum_{i,j}(\sum_{k=1}^{K}h^{j*}_{k2}h^{i}_{k2})\bm{G}^{j\dagger}\bm{G}^{-1}\bm{G}^{i} & \dots & \sum_{i,j}(\sum_{k=1}^{K}h^{j*}_{k2}h^{i}_{kL})\bm{G}^{j\dagger}\bm{G}^{-1}\bm{G}^{i} \\
    \vdots & \vdots & \ddots & \vdots \\
    \sum_{i,j}(\sum_{k=1}^{K}h^{j*}_{kL}h^{i}_{k1})\bm{G}^{j\dagger}\bm{G}^{-1}\bm{G}^{i} & \sum_{i,j}(\sum_{k=1}^{K}h^{j*}_{kL}h^{i}_{k2})\bm{G}^{j\dagger}\bm{G}^{-1}\bm{G}^{i} & \dots & \sum_{i,j}(\sum_{k=1}^{K}h^{j*}_{kL}h^{i}_{kL})\bm{G}^{j\dagger}\bm{G}^{-1}\bm{G}^{i} 
    \end{matrix} \right] \label{eqn:92} 
\end{align}
\end{figure*}
where $\bm{G}^j$ was defined in \eqref{eqn:defGj} and $i,j=0,\dots,J-1$. Note that, $\bm{G_H}^\dagger(\bm{I}_K\otimes\bm{G}^{-1})\bm{G_H}$ is an asymptotically block Toeplitz matrix. In order to find its generating matrix, we need to find the generating function of  $\bm{G}^{j\dagger}\bm{G}^{-1}\bm{G}^{i}$. According to {\cite{gray}}, this can be found by multiplying the generating functions of each of the three matrices. To find the generating function of $\bm{G}$, let's define the discrete time Fourier transform of $g[n]$ as
	\begin{eqnarray}
		G_d(f_n)  \triangleq   \sum_{p=-\infty}^{+\infty}g[p]e^{-j2\pi f_n p} 
		 = \frac{1}{\delta T}\sum_{m=-\infty}^{+\infty}G\left(\frac{f_n-m}{\delta T}\right), \label{eqn:GdGrelation}
	\end{eqnarray}
	 where $G(f)$ is the continuous time Fourier transform of $g(t)$. 
	 It is easy to see that $G_d(-f_n)$ is the generating function of matrix $\bm{G}$ according to the definition of generating function, but since $G_d(f_n)$ is an even function $G_d(-f_n)= G_d(f_n)$. The folded spectrum $G_d(f_n)$ is periodic in $f_n$ with period 1. Thus, 
\begin{align}
    &\mathcal{G}(\bm{G}^{j\dagger}\bm{G}^{-1}\bm{G}^{i})=\mathcal{G}(\bm{G}^{j\dagger})\mathcal{G}(\bm{G}^{-1})\mathcal{G}(\bm{G}^{i})\\
    &=\left(\sum_ng\left(-(n+j)\delta T\right)e^{j2\pi f_nn}\right) \notag\\ & \qquad\qquad\times\left(\sum_mg\left((m-i)\delta T\right)e^{j2\pi f_nm}\right)G_d^{-1}(f_n)\\
    &=G_d(f_n)e^{-j2\pi f_nj}G_d(f_n)e^{j2\pi f_ni}G_d^{-1}(f_n) \\
    &=G_d(f_n)e^{-j2\pi f_n(j-i)},
\end{align}
where we let $\mathcal{G}$ represent both finding the generating function and finding the generating matrix operation. 
	 With the derivation above we can easily obtain the generating matrix of $\bm{G_H}^\dagger(\bm{I}_K\otimes\bm{G}^{-1})\bm{G_H}$, whose $(n,m)$th entry is
\begin{align}
    &\big(\mathcal{G}(\bm{G_H}^\dagger(\bm{I}_K\otimes\bm{G}^{-1})\bm{G_H})\big)_{n,m}=\sum_{i,j}\left(\sum_{k=1}^{K}h^{j*}_{kn}h^{i}_{km}\right)\mathcal{G}(\bm{G}^{j\dagger}\bm{G}^{-1}\bm{G}^{i}) \\
    &=\sum_{k=1}^{K}\left[\left(\sum_{i=0}^{J-1}h^{i}_{km}e^{j2\pi f_ni}\right)\left(\sum_{j=0}^{J-1}h^{j*}_{km}e^{-j2\pi f_nj}\right)G_d(f_n)\right]  \\
    &=\sum_{k=1}^{K}H^*_{kn}(-f_n)H_{km}(-f_n)G_d(f_n).
\end{align}
Here $H_{km}(-f_n)$ is defined as $H_{km}(-f_n)=\sum_{i=0}^{J-1}h^{i}_{km}e^{j2\pi f_ni}$, which is the discrete time Fourier transform of the coefficients of the $(n,m)$th antenna link. We call this the link spectrum. Now we define the matrix 
\begin{equation}
    \tilde{\bm{H}}(-f_n) = \left[\begin{matrix} 
    H_{11}(-f_n)& H_{12}(-f_n)& \dots &H_{1L}(-f_n)\\ H_{21}(-f_n)& H_{22}(-f_n)& \dots &H_{2L}(-f_n) \\
    \vdots & \vdots & \ddots & \vdots \\
    H_{K1}(-f_n)& H_{K2}(-f_n)& \dots &H_{KL}(-f_n) 
    \end{matrix} \right]
\end{equation}
as the channel spectrum matrix, which is composed of the link spectrums. Then it is straightforward to write 
\begin{align}
    \mathcal{G}(\bm{G_H}^\dagger(\bm{I}_K\otimes\bm{G}^{-1})\bm{G_H})&=G_d(f_n)\tilde{\bm{H}}(-f_n)^\dagger\tilde{\bm{H}}(-f_n)\notag\\
    &=G_d(f_n)\tilde{\bm{Z}}(-f_n),
\end{align}
where $\tilde{\bm{Z}}(-f_n)\triangleq\tilde{\bm{H}}(-f_n)^\dagger\tilde{\bm{H}}(-f_n)$.

We also need to find the generating matrix for $\bm{\Sigma_A}$ to obtain $\mathcal{G}(\bm{\Sigma_A}\bm{G_H}^\dagger(\bm{I}_K\otimes\bm{G}^{-1})\bm{G_H})$, which is in the resulting expression after applying Szeg\"o's theorem on \eqref{eqn:cfsmanipulated}. As $\bm{\Sigma_A}$ is also a  block Toeplitz matrix, we define
\begin{align}
    \tilde{\bm{S}}_a(f_n) &= \mathcal{G}(\bm{\Sigma_A}) \notag\\
    &= \left[\begin{matrix} 
    S_{a,11}(f_n)& S_{a,12}(f_n)& \dots &S_{a,1L}(f_n)\\ S_{a,21}(f_n)& S_{a,22}(f_n)& \dots &S_{a,2L}(f_n)\\
    \vdots & \vdots & \ddots & \vdots \\
    S_{a,L1}(f_n)& S_{a,L2}(f_n)& \dots &S_{a,LL}(f_n) 
    \end{matrix} \right],
\end{align}
where $S_{a,ij}(f_n)=\mathcal{G}(\bm{K}_{\bm{a},ij}), i,j=1,\dots,L$. We call $\tilde{\bm{S}}_a(f_n)$ as the input generating matrix, the diagonal elements of which are the PSDs of input data sequences $a_i[n], i=1,\dots,L$, and the off-diagonal elements are the cross PSDs of input data sequences $a_i[n]$ and $a_j[n], i,j=1,\dots,L$, $i \neq j$. By applying \eqref{eqn:genszego} of Lemma \ref{lem:blktoep} on 
\eqref{eqn:cfsmanipulated} and on the power constraint {\eqref{eqn:fsfcdetpowconst}}, the optimization problem becomes
\begin{align}
    &C^{FS}=\underset{\tilde{\bm{S}}_a(f_n)}{\max}\int_{-\frac{1}{2}}^{\frac{1}{2}}\log_2\det\bigg[\bm{I}_{L}+\frac{G_d(f_n)}{\sigma^{2}_0}\tilde{\bm{S}}_a(f_n)\tilde{\bm{Z}}(f_n)\bigg]df_n \label{eqn:fsfcobj}\\
    & s.t. \quad\quad \frac{1}{\delta T}\int_{-\frac{1}{2}}^{\frac{1}{2}}\text{tr}\big[G_d(f_n)\tilde{\bm{S}}_a(f_n)\big]df_n\leq P.\label{eqn:fsfccons}
\end{align}
Before we solve the above optimization problem, let's define  $\tilde{\bm{S}}(f_n)  \triangleq G_d(f_n)\tilde{\bm{S}}_a(f_n)$.  
We call $\tilde{\bm{S}}(f_n)$ as the power generating matrix, since the integration over the sum of its diagonal elements is the transmit power.
Note that, \eqref{eqn:fsfcobj} contains both the influence from the FS channel $\tilde{\bm{Z}}(f_n)$ and the spectrum of the pulse shaping filter $G_d(f_n)$. This is because the ISI caused by FTN and the channel-induced ISI jointly affect the system. It can be considered that the channel has additional frequency selectivity with spectrum $G_d(f_n)$, which is due to FTN signaling.

We would like to emphasize that $\tilde{\bm{S}}(f_n)$ and $\tilde{\bm{Z}}(f_n)$ are always Hermitian matrices regardless of $f_n$. Therefore, we can use eigenvalue decomposition to diagonalize them and write 
\begin{align}
    \tilde{\bm{S}}(f_n)&=G_d(f_n)\tilde{\bm{S}}_a(f_n)=\tilde{\bm{U}}(f_n)\tilde{\bm{\Phi}}(f_n)\tilde{\bm{U}}(f_n)^\dagger \label{eqn:decompS}\\
    \tilde{\bm{Z}}(f_n)&=\tilde{\bm{V}}(f_n)\tilde{\bm{T}}(f_n)\tilde{\bm{V}}(f_n)^\dagger,\label{eqn:decompW}
\end{align}
where the diagonal matrices $\tilde{\bm{\Phi}}(f_n)$ and $\tilde{\bm{T}}(f_n)$ have the structure $\tilde{\bm{\Phi}}(f_n)=\text{diag}\{\phi_1(f_n), \phi_2(f_n),\dots, \phi_L(f_n)\}$ and $\tilde{\bm{T}}(f_n)=\text{diag}\{\tau_1(f_n), \tau_2(f_n),\dots, \tau_L(f_n)\}$. 
The channel matrix $\tilde{\bm{Z}}(f_n)$ is decomposed into eigenchannels, $\tau_i(f_n)$'s are the eigenmodes at frequency $f_n$, and $\phi_i(f_n)$'s are the power allocated to the eigenchannels, which we call the eigenspectrum. 

{
In the next step, we can upper bound \eqref{eqn:fsfcobj} using Hadamard inequality as
\begin{align}
    &\int_{-\frac{1}{2}}^{\frac{1}{2}}\log_2\det\left[\bm{I}_{L}+\sigma^{-2}_0\tilde{\bm{S}}(f_n)\tilde{\bm{Z}}(f_n)\right]df_n\leq \\
    &\qquad\qquad\qquad\int_{-\frac{1}{2}}^{\frac{1}{2}}\log_2\det\left[\bm{I}_{L}+\sigma^{-2}_0\tilde{\bm{\Phi}}(f_n)\tilde{\bm{T}}(f_n)\right]df_n,
\end{align}
and the upper bound is achieved when $\tilde{\bm{U}}(f_n)=\tilde{\bm{V}}(f_n),  \forall f_n$. Then the problem in \eqref{eqn:fsfcobj}-\eqref{eqn:fsfccons} becomes
\begin{align}
    &C^{FS}=\underset{\phi_1(f_n),\dots,\phi_L(f_n)}{\max}\int_{-\frac{1}{2}}^{\frac{1}{2}}\sum_{i=1}^{L}\log_2\left(1+\sigma^{-2}_0\phi_i(f_n)\tau_i(f_n)\right)df_n \label{eqn:fsfcobjdiag}\\
    & s.t. \quad\quad \frac{1}{\delta T}\int_{-\frac{1}{2}}^{\frac{1}{2}}\sum_{i=1}^{L}\phi_i(f_n)df_n\leq P.\label{eqn:fsfcconsdiag}
\end{align}
By applying the Lagrange multiplier method, we get the optimal solution for $\phi_i(f_n)$ as
\begin{equation}
    \phi_i(f_n)^*=\left(\frac{1}{\mu}-\frac{1}{\tau_i(f_n)}\right)^+, i=1,\dots,L, \quad f_n\in\left[-\frac{1}{2},\frac{1}{2}\right], \label{eqn:solufsfc}
\end{equation}
where superscript $*$ means optimal
. The Lagrange multiplier $\mu$ can be found by solving 
\begin{equation}
    \frac{1}{\delta T}\int_{-\frac{1}{2}}^{\frac{1}{2}}\sum_{i=1}^{L}\left(\frac{1}{\mu}-\frac{1}{\tau_i(f_n)}\right)^+df_n = P.
\end{equation} 
The optimal eigenspectrum for each eigenchannel is in the form of joint waterfilling in both spatial and spectral domains.  There are $L$ eigenchannels in total, and waterfilling is done according to the inverted eigenmodes $1/\tau_i(f_n)$. All eigenchannels have the same water level  $\frac{1}{\mu}$ such that the integrated power over all frequencies and all channels satisfies \eqref{eqn:fsfcconsdiag}.  This is illustrated in Fig. \ref{fig:fsfcwaterfill}. 
\begin{figure}
    \centering
    \includegraphics[scale=0.07]{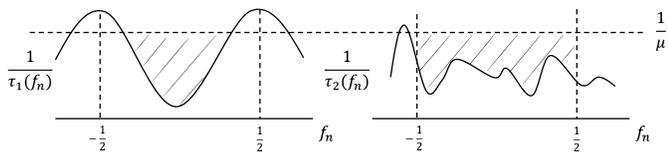}
    \caption{Waterfilling over a FS $2\times2$ MIMO channel with FTN. According to the inverted eigenmodes, the water level is the same for both eigenchannels.}
    \label{fig:fsfcwaterfill}
\end{figure}

Because we are applying FTN signaling, the resulting optimal generating matrix $\tilde{\bm{S}}^*_a(f_n)$ can be obtained 
according to \eqref{eqn:decompS}  as $\tilde{\bm{S}}^*_a(f_n)=\frac{\tilde{\bm{S}}^*(f_n)}{G_d(f_n)}. $ 
This is equivalent to dividing $\phi_i(f_n)$'s of $\tilde{\bm{\Phi}}(f_n)$ in \eqref{eqn:decompS} by $G_d(f_n)$, and we call this shaped power allocation $\tilde{\phi}_i(f_n)=\frac{\phi_i(f_n)}{G_d(f_n)}$ 
as the input eigenspectrum. 
This also shows the direct influence of FTN on the power allocation. When we design the input distribution, the power spectrum for each eigenchannel is shaped by $\frac{1}{G_d(f_n)}$.
\begin{remark}{\label{rem:6}} The influence of FTN on the power allocation is separate from the influence of the channel. The only difference FTN makes is further shaping the power spectrum by $\frac{1}{G_d(f_n)}$. Moreover, for orthogonal signaling with pulses following the Nyquist ISI criterion, the optimal input generating matrix is just $\tilde{\bm{S}}^*(f_n)$ since $G_d(f_n)=1$. 
\end{remark}

\subsection{Spectrum Analysis of Frequency Flat MIMO Channels with FTN}

In this section, we obtain the power spectrum of the frequency flat MIMO FTN channel as a special case of the above discussion. In frequency flat fading, the channel response has only one tap, $h_{kl}^j=h_{kl}\delta[j]$, where $h_{kl}$ is the channel coefficient for link $(k,l)$ and $\delta[j]$ is the discrete Dirac delta function. Then \eqref{eqn:fsfcobj} and \eqref{eqn:fsfccons} become
\begin{align}
    &C=\underset{\tilde{\bm{S}}(f_n)}{\max}\int_{-\frac{1}{2}}^{\frac{1}{2}}\log_2\det\bigg[\bm{I}_{L}+\sigma^{-2}_0\tilde{\bm{S}}(f_n)\tilde{\bm{Z}}(f_n)\bigg]df_n \label{eqn:mimospecobj}\\
    & s.t. \quad\quad \frac{1}{\delta T}\int_{-\frac{1}{2}}^{\frac{1}{2}}\text{tr}\big[\tilde{\bm{S}}(f_n)\big]df_n\leq P.\label{eqn:mimospeccons}
\end{align} This optimization is solved by the same {spectrum waterfilling} strategy of the previous section, only that the eigenmodes $\tau_i(f_n)$'s   
have a constant spectrum on {{$[-\frac{1}{2},\frac{1}{2}]$}} with amplitude $\tau_i$'s. The resulting optimal $\phi^*_i(f_n)$'s also have constant  spectrum  according to \eqref{eqn:solufsfc}, and the amplitudes are identical to the solution in (47). It is easy to see that the optimal power generating matrix in \eqref{eqn:decompS} becomes 
\begin{equation}
    \tilde{\bm{S}}(f_n)=\bm{V}(f_n)\bm{\Lambda}_{\bar{\alpha}}\bm{V}^\dagger(f_n)=\bm{V_Z\Lambda}_{\bar{\alpha}}\bm{V_Z}^\dagger.
\end{equation}
However, the optimal input generating matrix is not constant, it is equal to 
\begin{equation}
    \tilde{\bm{S}}_a(f_n)=\frac{\bm{V_Z\Lambda}_{\bar{\alpha}}\bm{V_Z}^\dagger}{G_d(f_n)}. \label{eqn:116}
\end{equation} We note that this matrix is exactly the same as the generating matrix of the optimal input covariance matrix in {\eqref{eqn:lll}}. The optimal input eigenspectrum in this case will be $\frac{\bar{\alpha}_i}{G_d(f_n)}$.
\begin{remark}
Due to \eqref{eqn:116}, we can see that in flat fading MIMO channels, the influence of FTN on the spectrum can be separated from the spatial effects of MIMO. This is also a natural consequence of Remark~\ref{rem:6} for FS MIMO FTN channels.
\end{remark}
\begin{remark}
We can consider frequency flat SISO channel with FTN as a special case of the above result and easily obtain the result in \cite{ganji}. 
\end{remark}
}
}
	\section{Numerical Results}\label{sec8}
	
	In this section, we display FTN capacity for different $(\delta, \beta)$ pairs for both frequency flat and FS MIMO channels. We assume $g(t)$ is a raised cosine pulse and $T=0.01$. 
	{The SNR is computed as $\text{SNR}=P/\sigma_0^2$. For frequency flat fading cases {in Figs. \ref{fig8} - \ref{fig:mimoftn}}, channel coefficients follow a circularly symmetric complex Gaussian distribution with zero mean and unit variance. 
	All the simulation results are averaged over 1000 random channel realizations, unless otherwise stated.}

			\begin{figure}[t]
		\centering
		\includegraphics[scale=0.6]{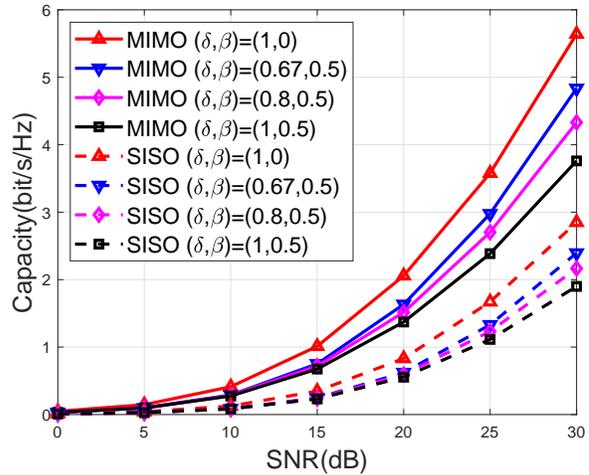}
		\caption{Frequency flat, $2\times 2$ MIMO FTN and SISO FTN capacity results for different $\delta$ values for the roll-off factor $\beta = 0.5$. The pairs $(\delta, \beta)=(1,0)$ and $(1,0.5)$ are respectively equivalent to Nyquist signaling with an ideal sinc pulse, and to Nyquist signaling with $\beta=0.5$.}
		\label{fig8}
	\end{figure}
	
			\begin{figure}[t]
		\centering
		\includegraphics[scale=0.6]{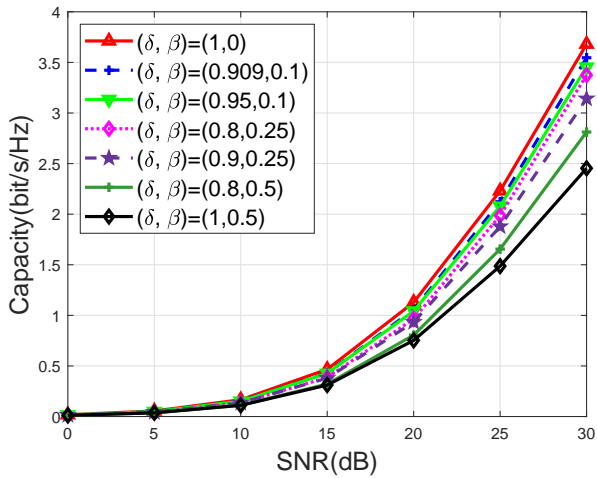}
		\caption{Frequency flat, $2\times 2$ MIMO FTN capacity for different $(\delta,\beta)$ values, displaying the effect of $\beta$ approaching 0. The pairs $(\delta, \beta)=(1,0)$ and $(1,0.5)$ are respectively equivalent to Nyquist signaling with an ideal sinc pulse, and to Nyquist signaling with $\beta=0.5$. }
		\label{taubeta}
	\end{figure}
	
	In Fig.~\ref{fig8} we compare the capacity of a frequency flat, $2\times 2$ MIMO FTN channel with the capacity of a frequency flat SISO FTN channel, for different acceleration factors $\delta$. The roll-off factor $\beta$ is set to $0.5$. We observe that for a given $\beta$, as we decrease $\delta$ capacity becomes larger for both MIMO and SISO. This spectral efficiency improvement clearly demonstrates the superiority of FTN for practical pulse shapes. FTN signaling utilizes the bandwidth more efficiently and outperforms orthogonal signaling. Despite the fact that FTN introduces ISI, increasing signaling rate or decreasing $\delta$ while $\delta(1+\beta)\geq 1$ satisfied, increases capacity. Note that, $(\delta,\beta)=(1,0)$; i.e. orthogonal signaling with ideal sinc pulses, is a performance upper bound for both SISO and MIMO FTN.

		\begin{figure}[t]
		\centering
		\includegraphics[scale=0.6]{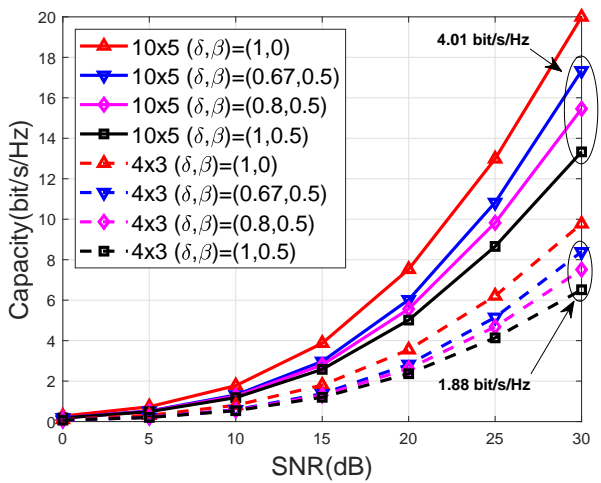}
		\caption{Frequency flat, $4\times 3$ and $10 \times 5$ MIMO FTN capacity results for different $\delta$ values for the roll-off factor $\beta = 0.5$. The pairs $(\delta, \beta)=(1,0)$ and $(1,0.5)$ are respectively equivalent to Nyquist signaling with an ideal sinc pulse, and to Nyquist signaling with $\beta=0.5$.}
		\label{fig9}
	\end{figure}
{In Fig.~\ref{taubeta}, we compare different $(\delta,\beta)$ pairs for the $2 \times 2$ MIMO FTN channel with frequency flat fading. We observe that the capacity increases as $\beta$ decreases, and for each $\beta$ the capacity increases as $\delta$ decreases while $\delta \geq 1/(1+\beta)$. Similar to the SISO case, we observe that for each $\beta$, the best $\delta = 1/(1+\beta)$, and overall the best $(\delta,\beta)$ pair is $(1,0)$. Although orthogonal signaling with sinc pulses achieve the best performance, practical systems operate at non-zero $\beta$ values and the acceleration factor should be chosen as $1/(1+\beta)$ to achieve the highest spectral efficiency possible for in MIMO FTN as well.}
	
We compare $4 \times 3$ and $10 \times 5$ MIMO FTN for frequency flat fading channels in Fig.~\ref{fig9}. We observe that there is a larger gain of FTN against Nyquist transmission when we use larger sized antenna arrays. For example, in Fig~\ref{fig9}, for $\beta=0.5$, changing $\delta$ from 1 to 0.67 brings in an extra $4.01$ bits/s/Hz in $10\times5$ MIMO FTN, while it offers only $1.88$ bits/s/Hz in $4\times3$ MIMO FTN signaling. Thus, the potential of FTN signaling increases as we increase the number of antennas.


	\begin{figure*}[htbp]
	    \centering
	    \includegraphics[scale=0.6]{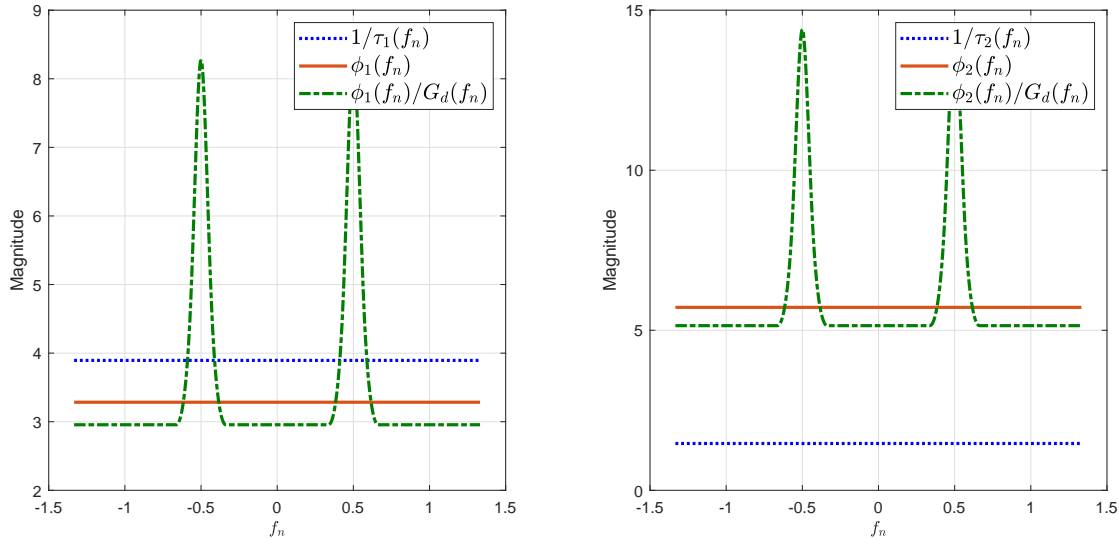}
	    \caption{Frequency domain power allocation for a $2\times2$ MIMO FTN in frequency flat channel.}
	    \label{fig:mimoftn}
	\end{figure*}
	{We plot the spectrum domain power allocation for a frequency flat, $2\times 2$ MIMO channel with FTN {for one particular channel realization} in Fig.~\ref{fig:mimoftn}. We assume $\delta=0.9$, $\beta=0.25$ and SNR is $10$dB. In the figure, the two eigenmodes $\tau_i$, $i= 1,2$, power allocation $\bar{\alpha}_i$ (or equivalently $\phi_i(f_n)$) and the optimal input eigenspectrum $\frac{\bar{\alpha}_i}{G_d(f_n)}$ (equivalently $\frac{\phi_i(f_n)}{G_d(f_n)}$) are shown. As we can see, the power allocation is composed of two independent parts, waterfilling in spatial domain and spectrum shaping in frequency domain as discussed in Remark~\ref{rem:6} in Section \ref{subsec:freqdomA}. Amplitude of $\phi_i(f_n)$ is bigger when eigenmode $\tau_i$ is stronger, which is a natural consequence of waterfilling.}
	
			\begin{figure}[t]
		\centering
		\includegraphics[scale=0.6]{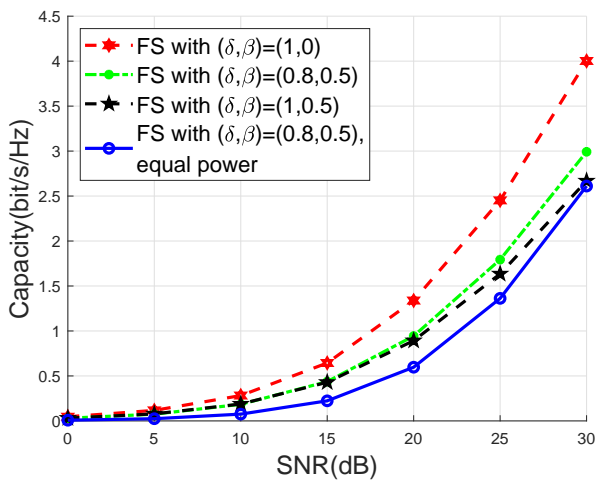}
		\caption{$2\times 2$ MIMO FTN spectral efficiency results under frequency selective fading for different $\delta$ values for the roll-off factor $\beta = 0.5$. The pairs $(\delta, \beta)=(1,0)$ and $(1,0.5)$ are respectively equivalent to Nyquist signaling with an ideal sinc pulse, and to Nyquist signaling with $\beta=0.5$. The blue curve with circle marks display the performance for equal power allocation. }
		\label{fig:fsfc}
	\end{figure}

\begin{figure}
    \centering
    \includegraphics[scale=0.6]{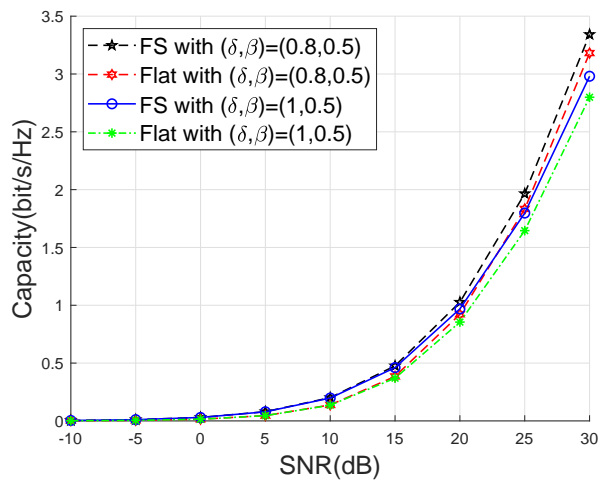}
    \caption{$2\times 2$ MIMO FTN capacity results under frequency flat and frequency selective fading channels for different $\delta$ values for the roll-off factor $\beta = 0.5$.}
    \label{fig:FTNnonFTNFsFlat}
\end{figure}
We plot the MIMO FTN spectral efficiency (SE) under FS fading for different $(\delta, \beta)$ pairs in Fig.~\ref{fig:fsfc}. We also compare the performance of equal power allocation with that of optimal power allocation. For equal power allocation, we assume $\bm{\Sigma_A}=c\bm{I}_{LN}$, where $c$ is a scaling factor that satisfies the power constraint \eqref{eqn:fsfcpowconst} with equality. As in  \eqref{eqn:fsfctimemodel}, we assume the gap between taps is $\delta T$, and all channel coefficients $h_{kl}^j$ are independent and identically distributed according to the complex Gaussian distribution $\mathcal{CN}(0,1/(LJ))$. We set the number of taps to $J=20$ and {$N=1000$}.  
As can be seen from the figure, FTN with optimum power allocation performs better than Nyquist signaling for the same $\beta$. The performance of Nyquist signaling with $(\delta,\beta)=(1,0)$ is again an upper bound on system performance. Moreover, when we compare optimal and equal power allocation for FTN signaling with $(\delta, \beta)=(0.8, 0.5)$, we observe that optimal power allocation can achieve $0.381$bit/s/Hz more rate with respect to equal power allocation.

{We plot Fig.~\ref{fig:FTNnonFTNFsFlat} with the same assumptions as in Fig.~\ref{fig:fsfc} to further study the effect of frequency selectivity. 
{All curves are for optimized power allocation.} We can see that FTN improves the performance for both flat fading and FS channels, and the benefit of FTN is independent from the channel, whether it is FS or flat fading. 
{ Note that the superiority of FS channel capacity against flat fading is due to additional diversity gain \cite{goldsmith}. } Moreover, in support of {Remark~\ref{rem:6}}, the figure shows that FTN does not change the fact that frequency selectivity {increases} performance.

\begin{figure*}[t]
    \centering
    \includegraphics[scale=0.66]{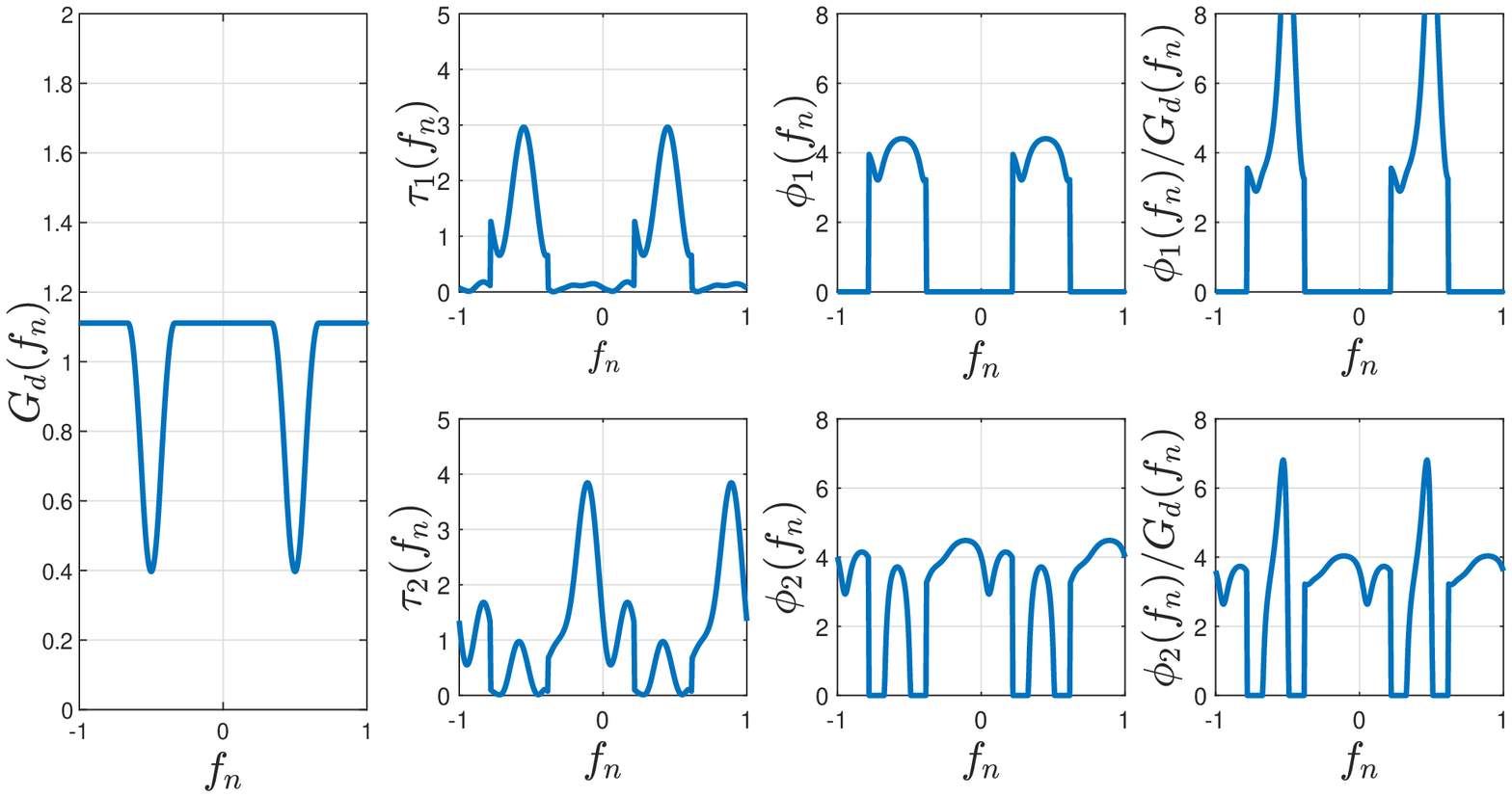}
    \caption{Spectrum domain analysis of  FTN signaling in $2\times2$ MIMO FS channel.}
    \label{fig:mimofsfc}
\end{figure*}

{In Fig. \ref{fig:mimofsfc}, we plot the spectrum domain power allocation of FS $2\times 2$ MIMO channel with FTN to show the joint influence of both the FS channel and FTN. We assume all channel coefficients follow i.i.d. Gaussian distribution with zero mean and variance $\frac{1}{LJ}$  
where $J=5$, SNR is $5$dB, $\delta=0.9$, and $\beta=0.25$.  The figures are plotted for one random channel realization only. 
All the spectrum in the subfigures are periodic with period 1.
As we can see, the water-filling result $\phi_1(f_n)$ and $\phi_2(f_n)$ of \eqref{eqn:solufsfc} have bigger amplitude where the eigenmodes $\tau_1(f_n)$ and $\tau_2(f_n)$ from \eqref{eqn:decompW} are strong. This is reasonable since it is not desirable to invest in weak channel conditions. The optimal input eigenspectrum $\frac{\phi_i(f_n)}{G_d(f_n)}$, $i= 1,2$, is amplified at where $G_d(f_n)$ has a dent, so that it will cancel the influence brought by FTN.}
}

\section{Conclusion} \label{sec9}

	\IEEEpeerreviewmaketitle

	
	%

	In this paper we investigate the capacity of FTN signaling for both frequency flat and frequency selective MIMO channels. The transmit power constraint takes the effect of inter-symbol interference into account, and is expressed in a matrix multiplication form. This expression simplifies the capacity optimization problem, and leads to the optimal joint precoding and waterfilling solution for MIMO FTN in flat fading channels. The optimal transmission scheme removes inter-symbol interference at the transmitter side and allocates power optimally over transmit antennas via waterfilling. 
	{We also study the capacity and the optimal transmission scheme for frequency selective fading channels. The capacity of MIMO FTN signaling in frequency selective channels can be achieved by water-filling in both spatial and spectrum domains together with spectrum shaping. Moreover, we prove that the effects of FTN and frequency selectivity on the optimized power spectrum are separable. The optimal input eigenspectrum functions are determined jointly by the channel and the influence of FTN; however, the {MIMO channel only affects the eigenspectrum functions, and the pulse shaping filter of FTN provides spectrum shaping.}
	Overall, we conclude that MIMO FTN significantly improves spectral efficiency with respect to orthogonal signaling for practical root raised cosine pulses with roll-off factor $\beta>0$ for both frequency flat and frequency selective channels. We believe that a complete analysis which accounts for $\delta(1+\beta)<1$, is a very complex problem, and it will be considered as future work. }
	

	\appendices	

	\section{}\label{appE}
	
	In this appendix, we prove that the matrix $\bm{\bar{\Sigma}_A}$ defined in \eqref{eqn:solumimo} as $(\bm{I}_L\otimes\bm{G})^{-1}\bar{\bm{K}}$ is a valid covariance matrix or equivalently it is both positive semidefinite and Hermitian. 
	
	We know that matrix $\bar{\bm{K}}$ is a Hermitian matrix. Moreover, $\bar{\bm{K}}$ is positive semidefinite, and its eigenvalues $\bar{\alpha}_i\geq 0, \forall i \in\{1,\cdots,NL\}$.
	
The matrix $(\bm{I}_L\otimes\bm{G})$ is also Hermitian. Moreover, due to Lemma \ref{lem:Gposemidef}, $\bm{G}$ is a positive semidefinite matrix. According to the property of the Kronecker product we have $(\bm{I}\otimes\bm{G})^{-1}=\bm{I}\otimes\bm{G}^{-1}$, which can be decomposed into $(\bm{Q}\otimes\bm{V_G})(\bm{I}_L\otimes\bm{\Lambda_G}^{-1})(\bm{Q}\otimes\bm{V_G})^\dagger$, where $\bm{Q}\in\mathbb{C}^{L\times L}$ is a unitary matrix composed of the eigenvectors of $\bm{I}_L$ \cite{matrix}. Furthermore, we can diagonalize the matrix $\bm{W}=(\bm{\tilde{H}}^\dagger\bm{\tilde{H}})\otimes\bm{I}_N$ of \eqref{eqn:W} as $(\bm{V}_{\bm{Z}}\otimes\bm{U})(\bm{\Lambda_Z}\otimes\bm{I}_N)(\bm{V}_{\bm{Z}}\otimes\bm{U})^\dagger$, where $\bm{U}\in\mathbb{C}^{N\times N}$ is a unitary matrix composed of the eigenvectors of $\bm{I}_N$.	As $\bm{Q}$ and $\bm{U}$ are composed of the eigenvectors of the identity matrix, they can be chosen as any unitary matrix. Therefore, they can also be chosen as  $\bm{Q}=\bm{V_Z}$ and  $\bm{U}=\bm{V_G}$. Thus, the optimal $\bar{\bm{K}}$ becomes
	\begin{equation}
		\bar{\bm{K}}=\bm{V_W}\bm{\bar{\Lambda}_K}\bm{V_W}^\dagger=(\bm{V}_{\bm{Z}}\otimes\bm{V_G})\bm{\bar{\Lambda}_K}(\bm{V}_{\bm{Z}}\otimes\bm{V_G})^\dagger.
	\end{equation}
	We can also write the matrix $\bm{I}\otimes\bm{G}^{-1}$ as 
	\begin{equation}
		(\bm{I}\otimes\bm{G}^{-1})=(\bm{V_Z}\otimes\bm{V_G})(\bm{I}_L\otimes\bm{\Lambda_G}^{-1})(\bm{V_Z}\otimes\bm{V_G})^\dagger.
	\end{equation}
	Then, the optimal solution for $\bm{\bar{\Sigma}_A}$ in \eqref{eqn:solumimo} can be simplified as 
	\begin{align}
		\bm{\bar{\Sigma}_A}&=(\bm{I}_L\otimes\bm{G})^{-1}\bm{\bar{K}} \\&=(\bm{V_Z}\otimes\bm{V_G})(\bm{I}_L\otimes\bm{\Lambda_G}^{-1})\bm{\bar{\Lambda}_K}(\bm{V}_{\bm{Z}}\otimes\bm{V_G})^\dagger. \label{eqn:eigoptsolu}
	\end{align}
	As $(\bm{I}_L\otimes\bm{\Lambda_G}^{-1})\bm{\bar{\Lambda}_K}$ is a diagonal matrix, we first observe that $\bm{\bar{\Sigma}_A}$ is a Hermitian matrix. Moreover, its eigenvalues are the diagonal elements of $(\bm{I}_L\otimes\bm{\Lambda_G}^{-1})\bm{\bar{\Lambda}_K}$.
	As the elements of $(\bm{I}_L\otimes\bm{\Lambda_G}^{-1})$ are strictly greater than zero, all entries of $(\bm{I}_L\otimes\bm{\Lambda_G}^{-1})\bm{\bar{\Lambda}_K}$ are also greater than or equal to zero. Thus, $\bm{\bar{\Sigma}_A}$ is also a positive semidefinite matrix. We conclude that $\bm{\bar{\Sigma}_A}$ is covariance matrix and the optimal solution found by  \eqref{eqn:solumimo} is a qualified solution.

	\section*{Acknowledgment}
	
	The authors would like to thank the editor and all three anonymous reviewers for all their valuable comments that improved the quality of the paper.

	\ifCLASSOPTIONcaptionsoff
	\newpage
	\fi

	\bibliographystyle{IEEEtran}
	\bibliography{main}

\begin{thebibliography}{10}
\providecommand{\url}[1]{#1}
\csname url@samestyle\endcsname
\providecommand{\newblock}{\relax}
\providecommand{\bibinfo}[2]{#2}
\providecommand{\BIBentrySTDinterwordspacing}{\spaceskip=0pt\relax}
\providecommand{\BIBentryALTinterwordstretchfactor}{4}
\providecommand{\BIBentryALTinterwordspacing}{\spaceskip=\fontdimen2\font plus
\BIBentryALTinterwordstretchfactor\fontdimen3\font minus
  \fontdimen4\font\relax}
\providecommand{\BIBforeignlanguage}[2]{{%
\expandafter\ifx\csname l@#1\endcsname\relax
\typeout{** WARNING: IEEEtran.bst: No hyphenation pattern has been}%
\typeout{** loaded for the language `#1'. Using the pattern for}%
\typeout{** the default language instead.}%
\else
\language=\csname l@#1\endcsname
\fi
#2}}
\providecommand{\BIBdecl}{\relax}
\BIBdecl

\bibitem{5gand6g}
T.~Nakamura, ``{5G evolution and 6G},'' in \emph{IEEE Symposium on VLSI
  Technology}, 2020, pp. 1--5.

\bibitem{mazo}
J.~E. Mazo, ``{Faster-than-{N}yquist signaling},'' \emph{The Bell System
  Technical Journal}, vol.~54, no.~8, pp. 1451--1462, 1975.

\bibitem{Cover}
T.~M. Cover and J.~A. Thomas, \emph{{Elements of Information Theory}}.\hskip
  1em plus 0.5em minus 0.4em\relax New York, NY, USA: Wiley-Interscience, 2006.

\bibitem{rusek}
F.~Rusek and J.~B. Anderson, ``{Constrained capacities for faster-than-Nyquist
  signaling},'' \emph{IEEE Transactions on Information Theory}, vol.~55, no.~2,
  pp. 764--775, 2009.

\bibitem{property}
Y.~J. Daniel~Kim, ``Properties of faster-than-{N}yquist channel matrices and
  folded-spectrum, and their applications,'' in \emph{IEEE Wireless
  Communications and Networking Conference (WCNC)}, 2016, pp. 1--7.

\bibitem{bajcsy}
Y.~J.~D. Kim and J.~Bajcsy, ``{Information rates of cyclostationary
  faster-than-Nyquist signaling},'' in \emph{12th Canadian Workshop on
  Information Theory (CWIT)}, 2011, pp. 1--4.

\bibitem{rusek12}
D.~Kapetanovic and F.~Rusek, ``{The effect of signaling rate on information
  rate for single carrier linear transmission systems},'' \emph{IEEE
  Transactions on Communications}, vol.~60, no.~2, pp. 421--428, 2012.

\bibitem{spectrashpingKramer}
M.~E. Hefnawy, G.~Dietl, and G.~Kramer, ``Spectral shaping for
  faster-than-{N}yquist signaling,'' in \emph{11th International Symposium on
  Wireless Communications Systems (ISWCS)}, 2014, pp. 496--500.

\bibitem{ganji}
M.~Ganji, X.~Zou, and H.~Jafarkhani, ``{On the capacity of faster than
  {N}yquist signaling},'' \emph{IEEE Communications Letters}, vol.~24, no.~6,
  pp. 1197--1201, 2020.

\bibitem{timelocalization}
A.~Gattami, E.~Ringh, and J.~Karlsson, ``Time localization and capacity of
  faster-than-{N}yquist signaling,'' in \emph{IEEE Global Communications
  Conference (GLOBECOM)}, 2015, pp. 1--7.

\bibitem{linearprecoding}
H.~Wang, A.~Liu, X.~Liang, S.~Peng, and K.~Wang, ``{Linear precoding for
  faster-than-Nyquist signaling},'' in \emph{3rd IEEE International Conference
  on Computer and Communications (ICCC)}, 2017, pp. 52--56.

\bibitem{svd}
T.~Ishihara and S.~Sugiura, ``{SVD}-precoded faster-than-{N}yquist signaling
  with optimal and truncated power allocation,'' \emph{IEEE Transactions on
  Wireless Communications}, vol.~18, no.~12, pp. 5909--5923, 2019.

\bibitem{eigendecomposition}
------, ``Eigendecomposition-precoded faster-than-nyquist signaling with
  optimal power allocation in frequency-selective fading channels,'' \emph{IEEE
  Transactions on Wireless Communications}, vol.~21, no.~3, pp. 1681--1693,
  2022.

\bibitem{chaki}
P.~Chaki, T.~Ishihara, and S.~Sugiura, ``Eigenvalue decomposition precoded
  faster-than-{N}yquist transmission of index modulated symbols,'' in
  \emph{IEEE International Symposium on Information Theory (ISIT)}, 2021, pp.
  3279--3284.

\bibitem{telatar}
E.~Telatar, ``{Capacity of multi-antenna {G}aussian channels},'' \emph{European
  Transactions on Telecommunications}, vol.~10, no.~6, pp. 585--595, 1999.

\bibitem{3gpp}
{3GPP, TR 25.814}, ``{Physical layer aspects for evolved universal terrestrial
  radio access (Release 7)}.''

\bibitem{rusek2009existence}
F.~Rusek, ``On the existence of the {M}azo-limit on {MIMO} channels,''
  \emph{IEEE Transactions on Wireless Communications}, vol.~8, no.~3, pp.
  1118--1121, 2009.

\bibitem{andrea}
A.~Modenini, F.~Rusek, and G.~Colavolpe, ``{Faster-than-{Nyquist signaling for
  next generation communication architectures}},'' in \emph{22nd European
  Signal Processing Conference (EUSIPCO)}, 2014, pp. 1856--1860.

\bibitem{michael}
M.~Yuhas, Y.~Feng, and J.~Bajcsy, ``{On the capacity of faster-than-{Nyquist
  MIMO} transmission with {CSI} at the receiver},'' in \emph{IEEE Globecom
  Workshops (GC Wkshps)}, 2015, pp. 1--6.

\bibitem{MIMOFTNfad}
S.~Wen, G.~Liu, F.~Xu, L.~Zhang, C.~Liu, and M.~A. Imran, ``Ergodic capacity of
  mimo faster-than-nyquist transmission over triply-selective rayleigh fading
  channels,'' \emph{IEEE Transactions on Communications}, vol.~70, no.~8, pp.
  5046--5058, 2022.

\bibitem{goldsmith}
A.~Goldsmith, \emph{{Wireless Communications}}.\hskip 1em plus 0.5em minus
  0.4em\relax Cambridge University Press, 2005.

\bibitem{measurecap}
P.~Almers, F.~Tufvesson, O.~Edfors, and A.~Molisch, ``{Measured capacity gain
  using water filling in frequency selective MIMO channels},'' in \emph{13th
  IEEE International Symposium on Personal, Indoor and Mobile Radio
  Communications (PIMRC)}, 2002, vol. 3, pp. 1347--1351.

\bibitem{covmat}
F.~Dupuy and P.~Loubaton, ``{On the capacity achieving covariance matrix for
  frequency selective MIMO channels using the asymptotic approach},''
  \emph{IEEE Transactions on Information Theory}, vol.~57, no.~9, pp.
  5737--5753, 2011.

\bibitem{finite}
A.~Scaglione, ``{Statistical analysis of the capacity of MIMO frequency
  selective Rayleigh fading channels with arbitrary number of inputs and
  outputs},'' in \emph{Proceedings IEEE International Symposium on Information
  Theory (ISIT)}, 2002, p. 278.

\bibitem{wangjui}
\BIBentryALTinterwordspacing
J.~T. Wang, ``Time–space mimo system with faster-than-{N}yquist signaling in
  frequency-selective fading channel,'' \emph{Multidimensional Syst. Signal
  Process.}, vol.~32, no.~3, pp. 1027--1039, Jul 2021. [Online]. Available:
  \url{https://doi.org/10.1007/s11045-021-00771-2}
\BIBentrySTDinterwordspacing

\bibitem{simo}
T.~E. Bogale, L.~B. Le, X.~Wang, and L.~Vandendorpe, ``Multipath multiplexing
  for capacity enhancement in {SIMO} wireless systems,'' \emph{IEEE
  Transactions on Wireless Communications}, vol.~16, no.~10, pp. 6895--6911,
  2017.

\bibitem{gray}
\BIBentryALTinterwordspacing
R.~M. Gray, ``Toeplitz and circulant matrices: A review,'' \emph{Foundations
  and Trends® in Communications and Information Theory}, vol.~2, no.~3, pp.
  155--239, 2006. [Online]. Available:
  \url{http://dx.doi.org/10.1561/0100000006}
\BIBentrySTDinterwordspacing

\bibitem{matrixbook}
F.~R. Gantmakler, \emph{\BIBforeignlanguage{eng}{The Theory of
  Matrices}}.\hskip 1em plus 0.5em minus 0.4em\relax New York: Chelsea
  Publishing Co., 1960.

\bibitem{timeseries}
P.~J. Brockwell, \emph{\BIBforeignlanguage{eng}{Time Series: {T}heory and
  Methods}}, 2nd~ed., ser. Springer Series in Statistics.\hskip 1em plus 0.5em
  minus 0.4em\relax New York: Springer-Verlag, 1991.

\bibitem{matrix}
R.~A. Horn, \emph{\BIBforeignlanguage{eng}{{Matrix Analysis}}}, 2nd~ed.\hskip
  1em plus 0.5em minus 0.4em\relax New York: Cambridge University Press, 2013.

\bibitem{spectrumbroaden}
Y.~J.~D. Kim and J.~Bajcsy, ``{On spectrum broadening of pre-coded
  faster-than-Nyquist signaling},'' in \emph{IEEE 72nd Vehicular Technology
  Conference - Fall}, 2010, pp. 1--5.

\bibitem{R3asymp}
Y.~G. Yoo and J.~H. Cho, ``{ Asymptotic optimality of binary
  faster-than-Nyquist signaling},'' \emph{IEEE Communications Letters},
  vol.~14, no.~9, pp. 788--790, 2010.

\bibitem{asynnoma}
\BIBentryALTinterwordspacing
M.~Ganji and H.~Jafarkhani, ``{Novel time asynchronous NOMA schemes for
  downlink transmissions},'' 2018. [Online]. Available:
  \url{https://arxiv.org/abs/1808.08665}
\BIBentrySTDinterwordspacing

\bibitem{blocktoep}
J.~Gutierrez-Gutierrez and P.~M. Crespo, ``{Asymptotically equivalent sequences
  of matrices and Hermitian block Toeplitz matrices with continuous symbols:
  Applications to MIMO systems},'' \emph{IEEE Transactions on Information
  Theory}, vol.~54, no.~12, pp. 5671--5680, 2008.

\end{thebibliography}

\end{document}